\documentclass[11pt]{amsart}

\usepackage{amssymb,latexsym, amsmath, amsxtra,amsfonts,amsthm}
\usepackage{cleveref}
\usepackage{caption}
\usepackage[usenames, dvipsnames]{color} 
\usepackage{dsfont}
\usepackage{empheq}
\usepackage{epsfig}
\usepackage{etoolbox}
\usepackage{float}
\usepackage{mathtools}
\usepackage{mleftright}
\usepackage{pgfplots}
\usepackage{subcaption}
\usepackage{tikz}
\usepackage{times}
\usepackage{url}
\usepackage{bm}
\usepackage[]{graphics}
\usepackage{bbm}
\usepackage[normalem]{ulem}
\usetikzlibrary{arrows,automata, pgfplots.fillbetween, matrix,positioning,fadings,calc,shapes,backgrounds,trees, plotmarks, patterns}

\newtheorem*{remark}{Remark}

\newtheorem{theorem}{Theorem}[section]

\newtheorem{lemma}[theorem]{Lemma}

\newtheorem{proposition}[theorem]{Proposition}
\newtheorem{definition}{Definition}[section]

\makeatletter
\newcommand*\bigcdot{\mathpalette\bigcdot@{.5}}
\newcommand*\bigcdot@[2]{\mathbin{\vcenter{\hbox{\scalebox{#2}{$\m@th#1\bullet$}}}}}
\makeatother
\makeatletter
\newcommand*\bbigcdot{\mathpalette\bigcdot@{.75}}
\newcommand*\bbigcdot@[2]{\mathbin{\vcenter{\hbox{\scalebox{#2}{$\m@th#1\bullet$}}}}}
\makeatother

\DeclareMathOperator{\Sgn}{sgn}

\DeclareMathOperator{\mom}{MoM}

\DeclareMathOperator{\ii}{\mathrm{i}\!}

\begin{document}

\title[Moments of moments of symplectic and orthogonal characteristic polynomials]{On the moments of the moments of the characteristic polynomials of Haar distributed symplectic and orthogonal matrices}
\abstract 
We establish formulae for the moments of the moments of the characteristic polynomials of random orthogonal and symplectic matrices in terms of certain lattice point count problems.  This allows us to establish asymptotic formulae when the matrix-size tends to infinity in terms of the volumes of certain regions involving continuous Gelfand-Tsetlin patterns with constraints.  The results we find differ from those in the unitary case considered previously.
\endabstract

\author{T. Assiotis}
\email{theo.assiotis@ed.ac.uk}
\address{School of Mathematics, University of Edinburgh, James Clerk Maxwell Building, Peter Guthrie Tait Rd, Edinburgh, EH9 3FD, United Kingdom}

\author{E. C. Bailey}
\email{ebailey@gc.cuny.edu}
\address{Department of Mathematics, CUNY Graduate Center, New York, 10016, United States of America}

\author{J. P. Keating}
\email{jon.keating@maths.ox.ac.uk}
\address{Mathematical Institute, University of Oxford, Oxford, OX2 6GG, United Kingdom}

\maketitle


\tableofcontents
\section{Introduction}\label{sec:intro} 

\subsection{Context}

Let
\begin{align*}
P_{G(N)}(\theta;g)=\det \left(I-ge^{-\ii \theta}\right)
\end{align*}
denote the characteristic polynomial on the unit circle (where $\ii\coloneqq\sqrt{-1}$) of a matrix $g \in G(N)$, for $G(N) \in \{ Sp(2N), SO(2N) \}$. Here, $Sp(2N)$ denotes the group of $2N \times 2N$ symplectic unitary matrices, and $SO(2N)$ denotes the group of $2N \times 2N$ orthogonal matrices and with determinant $+1$. We note that the eigenvalues of matrices from $Sp(2N)$ and $SO(2N)$ lie on the unit circle and come in complex conjugate pairs, namely they are of the form: $e^{\ii \phi_1}, e^{-\ii \phi_1}, e^{\ii \phi_2}, e^{-\ii \phi_2},\dots,e^{-\ii \phi_N}, e^{\ii \phi_N}$. In particular, we have that:
\begin{align}\label{ComplexConjObs}
\overline{P_{G(N)}(\theta;g)}=P_{G(N)}(-\theta;g).
\end{align}
 
Endowing the groups $Sp(2N)$ and $SO(2N)$ with the normalized Haar measure, we denote by $\mathbb{E}_{g \in G(N)}$ the mathematical expectation with respect to the corresponding measure on $G(N)$. We are interested in the following quantities, which we call the {\it moments of the moments} of the characteristic polynomial:
\begin{align}
\mom_{G(N)}\left(k,\beta\right)=\mathbb{E}_{g \in G(N)}\left[\left(\frac{1}{2\pi}\int_{0}^{2\pi}|P_{G(N)}(\theta;g)|^{2\beta}d\theta\right)^k\right].
\end{align}
Our focus will be on the asymptotics of $\mom_{G(N)}\left(k,\beta\right)$ in the limit as $N\to\infty$ when $k$ and $\beta$ are fixed integers.

When $G(N)$ is the unitary group $U(N)$, there has recently been a good deal of interest in the moments of the moments.  General conjectures were made concerning the large-$N$ asymptotics in this case by Fyodorov, Hiary and Keating in \cite{fyodorov12} and, in more detail, by Fyodorov and Keating in \cite{fyodorov14}.  These conjectures were explored in numerical computations and further generalized in \cite{FGK}.  One reason for studying the moments of the moments is that the conjectured asymptotics can be used to motivate conjectures for the extreme value statistics of the characteristic polynomials \cite{fyodorov12, fyodorov14}.   

In the case of the unitary group, the conjectured asymptotics for $\mom_{G(N)}\left(k,\beta\right)$ was proved when $k=2$ by Claeys and Krasovsky using a Riemann-Hilbert analysis \cite{CK}, and for all non-negative integer values of $k$ and $\beta$ by Bailey and Keating \cite{baikea18} using an approach based on exact formulae for finite $N$.  An alternative approach when $k$ and $\beta$ are non-negative integers was developed by Assiotis and Keating \cite{asskea19}, using a connection with representation theory and constrained Gelfand-Tsetlin patterns and thus establishing a connection with combinatorics.  This yields the same results as found in \cite{baikea18}, but leads to an alternative explicit formula for the coefficient appearing in the leading-order contribution to the asymptotics in terms of the volume of the associated Gelfand-Tsetlin polytopes; i.e.~it provides a geometrical interpretation for this constant.  Recently, Fahs has extended the approach developed in \cite{CK} to give a proof of the asymptotic formula for $\mom_{G(N)}\left(k,\beta\right)$ for non-negative integer values of $k$ and general non-negative real $\beta$, but without an explicit expression for the coefficient of the leading order term.  There is considerable interest in removing the assumption that $k$ is a non-negative integer though this is likely to require new ideas. Finally, there has also been a good deal of progress in proving the associated conjectures for the extreme value statistics of the characteristic polynomials; see, for example, \cite{ABB, CMN, PZ2017}.

Our purpose here is to extend the approach developed in \cite{asskea19} to give formulae for $\mom_{G(N)}\left(k,\beta\right)$, when $k$ and $\beta$ are non-negative integers and when $G(N)$ is either of the groups $Sp(2N)$ and $SO(2N)$, in terms of the associated constrained Gelfand-Tsetlin patterns (which are different to those that appear in the unitary case).  We then establish asymptotic formulae in which the volumes of the related Gelfand-Tsetlin polytopes appear.  Importantly, we find that the leading order asymptotic dependence on $N$ depends on the group in question.

We now have a well developed understanding of how to use results for random matrices to make conjectures about the corresponding questions in number theory.  For example, formulae for the moments of the moments of the characteristic polynomials of random unitary matrices, and for the extreme value statistics of the characteristic polynomials, can be used to motivate conjectures for the moments of the moments and for the extreme value statistics of the Riemann zeta-function on short intervals of its critical line \cite{fyodorov12, fyodorov14}.  There has recently been progress in proving these conjectures; see, for example, \cite{Najnudel, ABBRS, harper1, harper2}.  Our results here provide a similar basis for conjecturing formulae for the moments of the moments of $L$-functions from orthogonal and symplectic families, for example $L$-functions associated with quadratic twists of elliptic curves and quadratic Dirichlet $L$-functions, where the two averages are, first, over a short section of the critical line (e.g.~ a section of length 2$\pi$) centred on the symmetry point of the functional equation, and, second, over members of the family (i.e.~in the two examples given, over twists).  This application will be explored further in a subsequent paper.

It would be interesting to extend the approach developed in \cite{CK} and \cite{fahs19} to the orthogonal and symplectic groups.  This would require uniform asymptotics for determinants of the form Toeplitz + Hankel as the singularities merge; as far as we are aware this theory remains to be developed.  It would also be interesting to explore the implications of our results for orthogonal and symplectic analogues of Guassian Multiplicative Chaos, along the lines of the corresponding theory in the unitary case (see, for example, \cite{Webb, NSW}).

\subsection{Main results}

\begin{theorem}\label{MainTheoremSymplectic}
Let $G(N)=Sp(2N)$. Let $k,\beta \in \mathbb{N}$. Then, $\mom_{Sp(2N)}\left(k,\beta\right)$ is a polynomial function in $N$. Moreover, 
\begin{align}
\mom_{Sp(2N)}\left(k,\beta\right)=\mathfrak{c}_{Sp}(k,\beta)N^{k\beta(2k\beta+1)-k}+O\left(N^{k\beta(2k\beta+1)-k-1}\right),
\end{align}
where the leading order term coefficient $\mathfrak{c}_{Sp}(k,\beta)$ is the volume of a convex region defined in Section~\ref{sec:sympl_asympt} and is strictly positive.
\end{theorem}

\begin{theorem}\label{MainTheoremOrthogonal}
Let $G(N)=SO(2N)$.  Let $k, \beta\in\mathbb{N}$.  Then, $\mom_{SO(2N)}(k,\beta)$ is a polynomial function in $N$.  Moreover, 
\begin{align}
\mom_{SO(2N)}(1,1)&=2(N+1)\\
\intertext{otherwise,}
\mom_{SO(2N)}(k,\beta)&=\mathfrak{c}_{SO}(k,\beta)N^{k\beta(2k\beta-1)-k}+O\left(N^{k\beta(2k\beta-1)-k-1}\right),
\end{align}
where the leading order term coefficient $\mathfrak{c}_{SO}(k,\beta)$ is given as a sum of volumes of convex regions described in Section~\ref{sec:ortho_asympts} and is strictly positive. 
\end{theorem}

We remark that in the case of the unitary group, the power of $N$ appearing in the corresponding asymptotic formula is $k^2\beta^2-k+1$.

\subsection{Strategy of proof}

In order to prove our main results we combine the approaches that were developed in~\cite{baikea18} and~\cite{asskea19} (see also~\cite{krrr15}) for treating the simpler case of the unitary group. We first adapt an argument presented in~\cite{baikea18} to prove that $\mom_{G(N)}(k,\beta)$ is a polynomial in $N$. Then, in order to obtain the leading order term and an expression for its coefficient, we develop the combinatorial approach of~\cite{asskea19} to this setting.

The outline of the proof is as follows. We first obtain an expression for $\mom_{G(N)}(k,\beta)$ in terms of certain combinatorial objects, namely Gelfand-Tsetlin patterns, satisfying some (quite involved) constraints. We do this by making use of formulae due to Bump and Gamburd~\cite{bumgam06} that express averages of products of characteristic polynomials over the classical compact groups in terms of certain associated characters. The next step can be seen as taking a discrete to continuous limit, which gives the leading order coefficient as the volume of an explicit polytope, see Sections~\ref{sec:lattice}, ~\ref{sec:sympl_asympt}, and ~\ref{sec:ortho_asympts} for more precise statements.

There are certain important, not entirely technical, differences to the unitary group setting. In particular, the combinatorial objects we work with, namely the symplectic and orthogonal Gelfand-Tsetlin patterns, are more complicated than their unitary counterparts. For example, in order to apply the results required for the discrete to continuous limit in the orthogonal case, we first need to perform a decomposition of the corresponding patterns. The most significant difference however is the complexity of the constraints involved in the orthogonal and symplectic settings.  For the case of the unitary group, the constraints only depend on a single level of the pattern, whereas for the cases considered in this paper they involve several levels. 

This complication has the following consequences. Firstly, from the discrete to continuous limit argument it is not immediately clear that the leading order coefficient is actually strictly positive (which is straightforward in the unitary case). We manage to overcome this problem by a careful analysis of the different types of constraints. This is one of the more challenging parts of the paper, and the argument is supplemented by a number of diagrams. Secondly, the intricacies of the constraints prevents us, at least at present, from obtaining a more explicit expression for the leading order coefficient as was done in~\cite{asskea19} (such an expression has been used to connect this coefficient to Painlev\'e equations for $k=2$, see~\cite{krrr15} and~\cite{basgerub18}). However we do not believe that this is an intrinsic limitation of our approach, since, as we show in Section~\ref{sec:ks_asympts} for example, whenever such a leading order coefficient in an allied problem has been computed explicitly by different methods, it can fact also be reproduced by calculating volumes of Gelfand-Tsetlin polytopes.


\section{Preliminaries}\label{sec:prelims}  


\subsection{Symplectic and orthogonal Gelfand-Tsetlin patterns and Schur polynomials}\label{sec:gfdefs}

We will now give some background on symplectic and orthogonal Schur polynomials (which  are in fact Laurent polynomials). These can be defined as the characters of irreducible representations of the corresponding classical compact groups. From this perspective, making use of the Weyl character formula, one obtains well-known explicit expressions in terms of ratios of determinants (which we also record below). For our purposes however, we shall need some equivalent (see \cite{Proctor}) combinatorial definitions in terms of sums over objects called Gelfand-Tsetlin patterns. We mainly follow the recent exposition in Section 2 of \cite{ayyfis19}.

\begin{definition}[Signature]\label{def:signature}
A signature $\lambda$ of length $M$ is a sequence of $M$ non-increasing integers $(\lambda_1\ge \lambda_2 \ge \dots \ge \lambda_M)$. We denote the set of all such signatures by $\mathsf{S}_M$. We also denote the set of the signatures with non-negative entries by $\mathsf{S}_M^+$. For $\lambda=(\lambda_1,\dots,\lambda_M) \in \mathsf{S}_M^+$ we define $\lambda^-\coloneqq(\lambda_1,\dots,\lambda_{M-1},-\lambda_M)$. If $\lambda_1=\cdots=\lambda_M=n$ then we also write $\lambda=\langle n^M\rangle$.
\end{definition}

\begin{definition}[Interlacing]\label{def:interlacing}
We say that signatures $\lambda \in \mathsf{S}_M$ and $\nu \in \mathsf{S}_{M+1}$ interlace, and write $\lambda\prec \nu$, if:
\begin{align}\label{interlacing_eqs}
\nu_1\ge \lambda_1 \ge \nu_2 \ge \cdots \ge \nu_M \ge \lambda_M \ge \nu_{M+1}.
\end{align}
Similarly, we say that $\lambda \in \mathsf{S}_M$ and $\nu \in \mathsf{S}_M$ interlace, and still write $\lambda \prec \nu$ if:
\begin{align}
\nu_1\ge \lambda_1 \ge \nu_2 \ge \cdots \ge \nu_M \ge \lambda_M.
\end{align}
\end{definition}

We now define the notion of a half pattern, see Figure~\ref{fig:halfpattern} for an example. Symplectic and orthogonal Gelfand-Tsetlin patterns will be half patterns with additional properties.

\begin{definition}[Half patterns]\label{def:halfpat}
Let $n$ be a positive integer. A half (Gelfand-Tsetlin) pattern of length $n$ is given by a sequence of interlacing signatures $\left(\lambda^{(i)}\right)_{i=1}^{n}$ such that $\lambda^{(2i-1)},\lambda^{(2i)}\in \mathsf{S}_{i}$ and the interlacing is as follows:
\begin{align*}
\lambda^{(1)}\prec \lambda^{(2)}\prec \cdots \prec\lambda^{(n-1)}\prec \lambda^{(n)}.
\end{align*}
We call the first entries on the odd rows, namely $\lambda_{i}^{(2i-1)}$, the odd-starters.
\end{definition}

We arrive to the definition of a symplectic Gelfand-Tsetlin pattern, see Figure~\ref{fig:symplpattern} for an illustration.

\begin{definition}[Symplectic patterns]\label{def:symplpat}
Let $n$ be a positive integer. A $(2n)$-symplectic Gelfand-Tsetlin pattern $P=\left(\lambda^{(i)}\right)_{i=1}^{2n}$ is a half pattern of length $2n$ all of whose entries are non-negative integers. For fixed complex numbers $(x_1,\dots,x_n)$ we associate to the pattern $P$ a weight $w_{sp}(P)$ (dependence on $x_1,\dots,x_n$ is suppressed from the notation and will be clear from context in what follows) given by:
\begin{align*}
w_{sp}(P)=\prod_{i=1}^{n}x_i^{\sum_{j=1}^{i}\lambda_j^{(2i)}-2\sum_{j=1}^{i}\lambda_j^{(2i-1)}+\sum_{j=1}^{i-1}\lambda_j^{(2i-2)}},
\end{align*}
with $\lambda^{(0)}\equiv 0$. For $\nu \in \mathsf{S}_M^+$, we write $SP_{\nu}$ for the set of all $(2M)$-symplectic Gelfand-Tsetlin patterns with top row $\lambda^{(2M)}=\nu$. 
\end{definition}

We now give the combinatorial definition of the symplectic Schur polynomial as a sum of weights over symplectic patterns.
\begin{definition}[Symplectic Schur polynomial]\label{CombinatorialFormulaSymplectic}
Let $\nu \in \mathsf{S}_M^+$. We define the symplectic Schur polynomial by:
\begin{align}
sp^{(2M)}_{\nu}\left(x_1,\dots,x_M\right)=\sum_{P \in SP_{\nu}}^{}w_{sp}(P).
\end{align}
\end{definition}
It can be shown (see \cite{Proctor}) that this combinatorial definition coincides with the following determinantal form given by the Weyl character formula:
\begin{align*}
sp^{(2M)}_{\nu}\left(x_1,\dots,x_M\right)=\frac{\det\left(x_i^{\nu_j+M-j+1}-x_i^{-(\nu_j+M-j+1)}\right)_{i,j=1}^M}{\det\left(x_i^{M-j+1}-x_i^{-(M-j+1)}\right)_{i,j=1}^M}.
\end{align*}

We move on to the definition of orthogonal patterns. This is slightly more involved than the symplectic case since some of the elements are now permitted to be negative.  We will use the notation
\begin{align*}
\textnormal{sgn}(x)=
\begin{cases} +1,& x\ge 0 \\
-1,&  x<0.
\end{cases}
\end{align*}

\begin{definition}[Orthogonal patterns]\label{def:orthopat}
Let $n$ be a positive integer. A $(2n-1)$-orthogonal Gelfand-Tsetlin pattern $P=\left(\lambda^{(i)}\right)_{i=1}^{2n-1}$ is a half pattern of length $2n-1$ all of whose entries are either all integers or all half-integers\footnote{It transpires that for our application the entries of $(2n-1)$-orthogonal Gelfand-Tsetlin patterns are always all integers.} and which moreover satisfy:
\begin{itemize}
\item All entries except odd-starters are non-negative.
\item The odd-starters satisfy $|\lambda_i^{(2i-1)}|\le \min \{\lambda_{i-1}^{(2i-2)},\lambda_i^{(2i)}\}$ for $i=2,\dots,n-1$ and moreover \hbox{$|\lambda_1^{(1)}|\le \lambda_1^{(2)}$} and $|\lambda_n^{(2n-1)}|\le \lambda_{n-1}^{(2n-2)}$.
\end{itemize}
For fixed complex numbers $(x_1,\dots,x_n)$ we associate to the pattern $P$ a weight $w_{o}(P)$ given by:
\begin{align*}
w_{o}(P)=\prod_{i=1}^{n}x_i^{\textnormal{sgn}(\lambda_i^{(2i-1)})\textnormal{sgn}(\lambda_{i-1}^{(2i-3)})\left[\sum_{j=1}^{i}|\lambda_j^{(2i-1)}|-2\sum_{j=1}^{i-1}|\lambda_j^{(2i-2)}|+\sum_{j=1}^{i-1}|\lambda_j^{(2i-3)}|\right]},
\end{align*}
with $\lambda^{(0)},\lambda^{(-1)}\equiv 0$. For $\nu \in \mathsf{S}_M$, we write $OP_{\nu}$ for the set of all $(2M-1)$-orthogonal Gelfand-Tsetlin patterns with top row $\lambda^{(2M-1)}=\nu$. 
\end{definition}

See Figure~\ref{fig:orthopattern} for an example of an orthogonal Gelfand-Tsetlin pattern. 

As in the symplectic case, we have the following combinatorial definition of the orthogonal Schur polynomial as a sum of weights over orthogonal patterns.

\begin{definition}[Orthogonal Schur polynomial]\label{def:orthoschur}
Let $\nu \in \mathsf{S}_M^+$. We define the orthogonal Schur polynomial by:
\begin{align}
o^{(2M)}_{\nu}\left(x_1,\dots,x_M\right)=\sum_{P \in OP_{\nu}\cup OP_{\nu^-}}^{}w_{o}(P).
\end{align}
\end{definition}
Again, it can be shown (see \cite{Proctor}) that this combinatorial definition coincides with the following determinantal expression given by the Weyl character formula:
\begin{align*}
o^{(2M)}_{\nu}\left(x_1,\dots,x_M\right)=\frac{2\det\left(x_i^{\nu_j+M-j}+x_i^{-(\nu_j+M-j)}\right)_{i,j=1}^M}{\det\left(x_i^{M-j}+x_i^{-(M-j)}\right)_{i,j=1}^M}.
\end{align*}

%

\begin{figure}[!htb]
\centering
\begin{tikzpicture}
\node at (0,0) {$\lambda^{(1)}_{1}$};
\node at (1,1) {$\lambda^{(2)}_{1}$};
\node at (0,2) {$\lambda^{(3)}_{2}$};
\node at (2,2) {$\lambda^{(3)}_{1}$};
\node at (1,3) {$\lambda^{(4)}_{2}$};
\node at (3,3) {$\lambda^{(4)}_{1}$};

\draw (0.5,0.5) node[rotate=45] {$\leq$};
\draw (0.5,1.5) node[rotate=-45] {$\leq$};
\draw (1.5,1.5) node[rotate=45] {$\leq$};
\draw (1.5,2.5) node[rotate=-45] {$\leq$};
\draw (2.5,2.5) node[rotate=45] {$\leq$};
\end{tikzpicture}

\caption{A half pattern of length 4, $(\lambda^{(i)})_{i=1}^4$, with the interlacing explicitly shown.}\label{fig:halfpattern}
\end{figure}
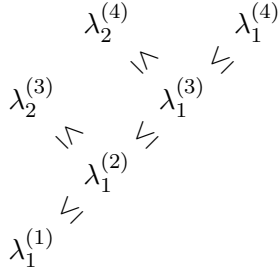

\begin{figure}[!htb]
\centering
\begin{subfigure}[t]{.48\textwidth}
\centering
\begin{tikzpicture}
\node at (0,0) {$1$};
\node at (1,1) {$2$};
\node at (0,2) {$1$};
\node at (2,2) {$2$};
\node at (1,3) {$2$};
\node at (3,3) {$3$};
\draw (-0.5,-0.5) -- (3,-0.5);
\node at (1,-1) {$w_{sp}(P)=x_2$};
\end{tikzpicture}
\caption{An example of a (4)-symplectic Gelfand-Tsetlin pattern $P$, with its corresponding weight $w_{sp}(P)$ below for some complex numbers $x_1, x_2$ as appearing in Definition~\ref{def:symplpat}.}\label{fig:symplpattern}
\end{subfigure}\hfill
\begin{subfigure}[t]{.48\textwidth}
\centering
\begin{tikzpicture}
\node at (0,0) {$-1$};
\node at (1,1) {$1$};
\node at (0,2) {$0$};
\node at (2,2) {$2$};
\node at (1,3) {$2$};
\node at (3,3) {$2$};
\node at (0,4) {$-2$};
\node at (2,4) {$2$};
\node at (4,4) {$4$};
\draw (-0.5,-0.5) -- (3,-0.5);
\node at (1.5,-1) {$w_{o}(P)=(x_1x_2x_3^2)^{-1}$};
\end{tikzpicture}
\caption{An example of a (5)-orthogonal Gelfand-Tsetlin pattern $P$, with its corresponding weight $w_{o}(P)$ below for some complex numbers $x_1, x_2, x_3$ as appearing in Definition~\ref{def:orthopat}.}\label{fig:orthopattern}
\end{subfigure}
\caption{Figures giving examples of symplectic and orthogonal Gelfand-Tsetlin patterns.}
\end{figure}
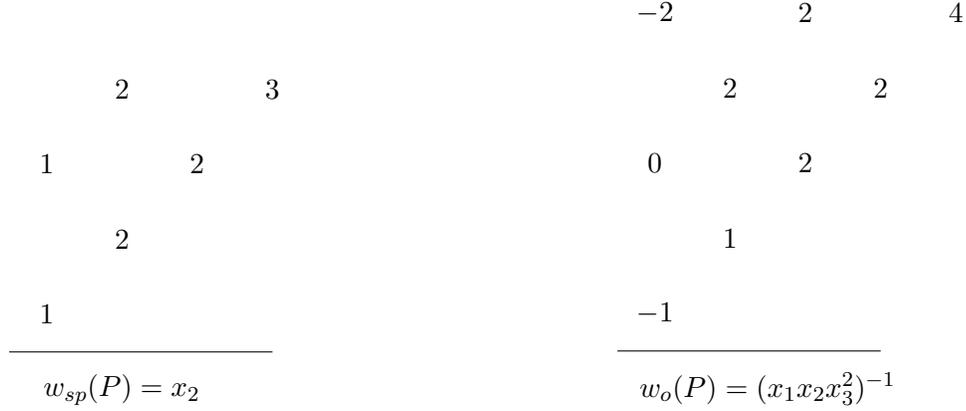


\subsection{Averages of products of characteristic polynomials as Schur polynomials}\label{sec:bgdefs} 

We have the following results due to Bump and Gamburd, see Sections 5 and 6 in \cite{bumgam06} (note that \cite{bumgam06} uses the equivalent definition of Schur polynomials in terms of determinants). These relate products of characteristic polynomials averaged (with respect to Haar measure) over the classical compact groups with Schur  polynomials.

\begin{proposition}\label{BumpGamburdSymplectic} Let $M$ be a positive integer and $x_1,\dots,x_M$ be complex numbers. Then,
\begin{align} 
\mathbb{E}_{g \in Sp(2N)}\left[\prod_{j=1}^{M}\det \left(I-x_jg\right)\right]&=(x_1\cdots x_M)^N sp^{(2M)}_{\langle N^M\rangle}\left(x_1,\dots,x_M\right).
\end{align}
\end{proposition}

\begin{proposition}\label{BumpGamburdOrthogonal} Let $M$ be a positive integer and $x_1,\dots,x_M$ be complex numbers. Then,
\begin{align}
\mathbb{E}_{g \in SO(2N)}\left[\prod_{j=1}^{M}\det \left(I-x_jg\right)\right]&=(x_1\cdots x_M)^N o^{(2M)}_{\langle N^M\rangle}\left(x_1,\dots,x_M\right).
\end{align}
\end{proposition}

%

In our applications below we will be taking particular choices of the complex numbers $x_1,\dots,x_M$ lying on the unit circle in the complex plane for some even integer $M$.


\subsection{Asymptotics of the number of lattice points in convex sets}\label{sec:lattice} 

We have the following theorem on the number of lattice points in convex regions of Euclidean space, see for example Section 2 in \cite{LatticePointCount}. 
\begin{theorem}\label{LatticePointCountTheorem}
Assume $\mathcal{S} \subset \mathbb{R}^{L}$ is a convex region contained in a closed ball of radius $\rho$. Then,
\begin{align}
\#\left(\mathcal{S} \cap \mathbb{Z}^L \right)=\textnormal{vol}_L\left(\mathcal{S}\right)+O_L\left(\rho^{L-1}\right),
\end{align}
where the implicit constant in the error term depends only on $L$.
\end{theorem}

We will prove our main results on the asymptotics of the moments of the moments by applying the theorem above with some judicious choices (different for each group) of the convex set $\mathcal{S}$.


\subsection{Averages of products of characteristic polynomials as combinatorial sums}\label{sec:permsum} 

Instead of expressing the averages of products of characteristic polynomials over the various matrix groups in terms of their Schur polynomials, one can instead view them as combinatorial sums.  These descriptions follow from work of Conrey et al.~\cite{cfkrs1} and will be used when determining the polynomial structure of the moments of moments. 

\begin{proposition}\label{combsumsympl} 
Let $M$ be a positive integer and $x_1,\dots,x_M$ be complex numbers. Then,
\begin{align*} 
\mathbb{E}_{g \in Sp(2N)}\left[\prod_{j=1}^{M}\det \left(I-x_jg\right)\right]&=(x_1\cdots x_M)^N\sum_{\varepsilon_j\in\{-1,1\}}\frac{\prod_{j=1}^Mx_j^{\varepsilon_j N}}{\prod_{1\leq i\leq j\leq M}(1-x_i^{-\varepsilon_i}x_j^{-\varepsilon_j})}
\end{align*}
\end{proposition}

\begin{proposition}\label{combsumorthog}
Let $M$ be a positive integer and $x_1,\dots,x_M$ be complex numbers. Then,
\begin{align*}
\mathbb{E}_{g \in SO(2N)}\left[\prod_{j=1}^{M}\det \left(I-x_jg\right)\right]&=(x_1\cdots x_M)^N\sum_{\varepsilon_j\in\{-1,1\}}\frac{\prod_{j=1}^Mx_j^{\varepsilon_j N}}{\prod_{1\leq i< j\leq M}(1-x_i^{-\varepsilon_i}x_j^{-\varepsilon_j})}
\end{align*}
\end{proposition}


Once more, $M$ will be an even integer and we will be picking the complex numbers $x_1,\dots,x_M$ in a particular way, always lying on the unit circle in the complex plane.


\section{Polynomial structure}\label{sec:poly} 

In this section we prove the following proposition.  This, together with results stated in Sections~\ref{sec:sympl} and ~\ref{sec:orthog} will prove Theorem \ref{MainTheoremSymplectic} and \ref{MainTheoremOrthogonal}.

\begin{proposition}\label{polystrucsympl}
Let $G(N)=Sp(2N)$, or $G(N)=SO(2N)$, and $k, \beta\in \mathbb{N}$.  Then $\mom_{G(N)}(k,\beta)$ is a polynomial function of $N$. 
\end{proposition}

\begin{proof}
We make use of the expressions for averages through the different matrix groups due to Conrey et al.~\cite{cfkrs1} that were introduced in section~\ref{sec:permsum}.  The argument follows that for the moments of the moments of the characteristic polynomials of unitary matrices, presented in~\cite{baikea18}.

We begin with the symplectic case. We apply Fubini's Theorem to obtain:
\begin{align}
\mom_{Sp(2N)}(k,\beta)=\frac{1}{(2\pi)^k}\int_0^{2\pi}\cdots\int_0^{2\pi}\mathbb{E}_{g \in Sp(2N)}\left[\prod_{j=1}^{2k\beta}\det \left(I-x_jg\right)\right]d\theta_1\cdots d\theta_k,
\end{align}
where, by recalling observation (\ref{ComplexConjObs}): \[\underline{x}=(\underbrace{e^{-\ii\theta_1},\dots,e^{-\ii\theta_1}}_\beta,\underbrace{e^{\ii\theta_1},\dots,e^{\ii\theta_1}}_\beta,\underbrace{e^{-\ii\theta_2},\dots,e^{-\ii\theta_2}}_\beta,\underbrace{e^{\ii\theta_2},\dots,e^{\ii\theta_2}}_\beta,\dots,\underbrace{e^{-\ii\theta_k},\dots,e^{-\ii\theta_k}}_\beta,\underbrace{e^{\ii\theta_k},\dots,e^{\ii\theta_k}}_\beta).\]
Then, by Proposition~\ref{combsumsympl}, we can write the moments of moments in the following form. 

\begin{equation*}
\mom_{Sp(2N)}(k,\beta)=\frac{1}{(2\pi)^k}\int_0^{2\pi}\cdots\int_0^{2\pi}\sum_{\varepsilon_j\in\{-1,1\}}\frac{\prod_{j=1}^{2k\beta}x_j^{\varepsilon_j N}}{\prod_{1\leq i\leq j\leq 2k\beta}(1-x_i^{-\varepsilon_i}x_j^{-\varepsilon_j})}
d\theta_1\cdots d\theta_k.
\end{equation*}

Above, each summand appears to have a pole of finite order (when $x_i^{\varepsilon_i}=x_j^{-\varepsilon_j}$), but these cancel with zeros in the numerator when the sum is considered as a whole.  This is clearly the case since the average of a product of polynomials is bounded~\cite{cfkrs1}.  Following this calculation, one may compute the resulting function by applying  l'H\^opital's rule a finite number of times, which results in a polynomial function in the variables $e^{\ii\theta_1},\dots,e^{\ii\theta_k}$, and whose coefficients are themselves polynomials in $N$.  Finally, after performing the integration over the $\theta_1,\dots,\theta_k$, only the constant term of said polynomial survives, which as noted is a polynomial in $N$. This concludes the proof of Proposition~\ref{polystrucsympl}.  The argument for the orthogonal case is completely analogous via Proposition~\ref{combsumorthog}.

\end{proof}


\section{Results for the symplectic group $Sp(2N)$}\label{sec:sympl}  

Here we give the proof of the leading order behaviour and coefficient of $\mom_{Sp(2N)}(k,\beta)$ as described in Theorem~\ref{MainTheoremSymplectic}.  The argument is split in to stages.  Firstly, we give an expression for the moments of moments using symplectic Gelfand-Tsetlin patterns with constraints.  Secondly, we observe that part of the pattern is determined, and hence only the `free' part plays a role.  Finally, by essentially passing from a discrete to a continuous setting and using the results presented in Section~\ref{sec:lattice}, we arrive at the result.  

\subsection{A combinatorial representation}
We begin with a combinatorial representation for $\mom_{Sp(2N)}(k,\beta)$.
\begin{proposition}\label{CombRepSymp1}
Let $k,\beta \in \mathbb{N}$. Then, $\mom_{Sp(2N)}(k,\beta)$ is equal to the number of $(4k\beta)$-symplectic Gelfand-Tsetlin patterns $P=\left(\lambda^{(i)}\right)_{i=1}^{4k \beta}$ with top row $\lambda^{(4k\beta)}=\langle N^{2k\beta}\rangle$, which moreover satisfy the following $k$ constraints for $i=1,\dots,k$:
\begin{align}
\sum_{j=(2i-2)\beta+1}^{(2i-1)\beta}\left[\sum_{l=1}^{j}\lambda_l^{(2j)}-2\sum_{l=1}^{j}\lambda_l^{(2j-1)}+\sum_{l=1}^{j-1}\lambda_l^{(2j-2)}\right]=\sum_{j=(2i-1)\beta+1}^{2i\beta}\left[\sum_{l=1}^{j}\lambda_l^{(2j)}-2\sum_{l=1}^{j}\lambda_l^{(2j-1)}+\sum_{l=1}^{j-1}\lambda_l^{(2j-2)}\right].
\end{align}
We denote the set of such patterns by $GT_{Sp}(N;k;\beta)$.
\end{proposition}

\begin{proof}
As in Proposition \ref{polystrucsympl}, by an application of Fubini's Theorem we have:
\begin{align}
\mom_{Sp(2N)}(k,\beta)=\frac{1}{(2\pi)^k}\int_0^{2\pi}\cdots\int_0^{2\pi}\mathbb{E}_{g \in Sp(2N)}\left[\prod_{j=1}^{2k\beta}\det \left(I-x_jg\right)\right]d\theta_1\cdots d\theta_k,\label{mom_sympl_fubini}
\end{align}
with (using (\ref{ComplexConjObs})) \[\underline{x}=(\underbrace{e^{-\ii\theta_1},\dots,e^{-\ii\theta_1}}_\beta,\underbrace{e^{\ii\theta_1},\dots,e^{\ii\theta_1}}_\beta,\underbrace{e^{-\ii\theta_2},\dots,e^{-\ii\theta_2}}_\beta,\underbrace{e^{\ii\theta_2},\dots,e^{\ii\theta_2}}_\beta,\dots,\underbrace{e^{-\ii\theta_k},\dots,e^{-\ii\theta_k}}_\beta,\underbrace{e^{\ii\theta_k},\dots,e^{\ii\theta_k}}_\beta).\]

Now, we make use of Proposition \ref{BumpGamburdSymplectic} along with Definition~\ref{CombinatorialFormulaSymplectic} to rewrite the integrand in~\eqref{mom_sympl_fubini} as follows, where the signature determining the set $SP_\nu$ is $\nu=\langle N^{2k\beta} \rangle \in \mathsf{S}_{2k\beta}^+$.
\begin{align*}
\mathbb{E}&_{g \in Sp(2N)}\left[\prod_{j=1}^{2k\beta}\det \left(I-x_jg\right)\right]=\\
&\sum_{P \in SP_{\langle N^{2k\beta} \rangle}}^{} \prod_{j=1}^{\beta}e^{-\ii \theta_1 \left[\sum_{l=1}^{j}\lambda_l^{(2j)}-2\sum_{l=1}^{j}\lambda_l^{(2j-1)}+\sum_{l=1}^{j-1}\lambda_l^{(2j-2)}\right]}\prod_{j=\beta+1}^{2\beta}e^{\ii \theta_1 \left[\sum_{l=1}^{j}\lambda_l^{(2j)}-2\sum_{l=1}^{j}\lambda_l^{(2j-1)}+\sum_{l=1}^{j-1}\lambda_l^{(2j-2)}\right]}\\
&\times \prod_{j=2\beta+1}^{3\beta}e^{-\ii \theta_2 \left[\sum_{l=1}^{j}\lambda_l^{(2j)}-2\sum_{l=1}^{j}\lambda_l^{(2j-1)}+\sum_{l=1}^{j-1}\lambda_l^{(2j-2)}\right]}\prod_{j=3\beta+1}^{4\beta}e^{\ii \theta_2 \left[\sum_{l=1}^{j}\lambda_l^{(2j)}-2\sum_{l=1}^{j}\lambda_l^{(2j-1)}+\sum_{l=1}^{j-1}\lambda_l^{(2j-2)}\right]} \times\cdots\\
&\times\prod_{j=(2k-2)\beta+1}^{(2k-1)\beta}e^{-\ii \theta_k \left[\sum_{l=1}^{j}\lambda_l^{(2j)}-2\sum_{l=1}^{j}\lambda_l^{(2j-1)}+\sum_{l=1}^{j-1}\lambda_l^{(2j-2)}\right]}\prod_{j=(2k-1)\beta+1}^{2k\beta}e^{\ii \theta_k \left[\sum_{l=1}^{j}\lambda_l^{(2j)}-2\sum_{l=1}^{j}\lambda_l^{(2j-1)}+\sum_{l=1}^{j-1}\lambda_l^{(2j-2)}\right]} .
\end{align*}
Finally, by making use of the fact that,
\begin{align*}
\frac{1}{2\pi}\int_{0}^{2\pi} e^{\ii s \theta}d\theta=\delta_{s=0},
\end{align*}
the statement of the proposition readily follows.
\end{proof}

We now make the simple observation that the form of the top signature $\langle N^{2k\beta} \rangle$ essentially fixes the top right triangle of a pattern in $GT_{Sp}(N;k;\beta)$, see Figure~\ref{fig:symplfixed}. In order to formalize the argument, it is convenient to have the following definition:

\begin{definition}\label{def:sympl_int_array}
Consider the following set of integer arrays \hbox{$\left(y^{(i)}\right)_{i=1}^{4k\beta-1} \in \mathbb{Z}^{k\beta(2k\beta+1)}$}, which we denote by $\mathfrak{I}_{Sp}(N;k;\beta)$, and which additionally satisfy the following conditions,
\begin{enumerate}
\item for all $1\leq i\leq 2k\beta$, $y^{(i)},y^{(4k\beta-i)} \in \mathsf{S}_{\lfloor\frac{i+1}{2} \rfloor}^+$,
\item both $\left(y^{(i)}\right)_{i=1}^{2k\beta}$ and $\left(y^{(4k\beta-i)}\right)_{i=1}^{2k\beta}$ form $(2k\beta)$-symplectic Gelfand-Tsetlin patterns,
\item $0\le y_j^{(i)} \le N$ for any valid $i, j$,
\item the rows $\left(y^{(i)}\right)_{i=1}^{4k\beta-1}$ fulfil the following constraints: 

In the case $k$ is even, let $i=1,\dots,\frac{k}{2}$ (with $y^{(0)},y^{(4k\beta)}\equiv 0$).  Then,
\begin{align}\label{eq:sympl_gf_const1}
\sum_{j=(2i-2)\beta+1}^{(2i-1)\beta}\left[\sum_{l=1}^{j}y_l^{(2j)}-2\sum_{l=1}^{j}y_l^{(2j-1)}+\sum_{l=1}^{j-1}y_l^{(2j-2)}\right]=\sum_{j=(2i-1)\beta+1}^{2i\beta}\left[\sum_{l=1}^{j}y_l^{(2j)}-2\sum_{l=1}^{j}y_l^{(2j-1)}+\sum_{l=1}^{j-1}y_l^{(2j-2)}\right],
\end{align}
and
\begin{align}\label{eq:sympl_gf_const2}
\sum_{j=(2i-2)\beta+1}^{(2i-1)\beta}\left[\sum_{l=1}^{j}y_l^{(4k\beta-2j)}-2\sum_{l=1}^{j}\right.&\left.y_l^{(4k\beta-2j+1)}+\sum_{l=1}^{j-1}y_l^{(4k\beta-2j+2)}\right]\\
&\qquad=\sum_{j=(2i-1)\beta+1}^{2i\beta}\left[\sum_{l=1}^{j}y_l^{(4k\beta-2j)}-2\sum_{l=1}^{j}y_l^{(4k\beta-2j+1)}+\sum_{l=1}^{j-1}y_l^{(4k\beta-2j+2)}\right].\nonumber
\end{align}
While, when $k$ is odd we have the same constraints as above for $i=1,\dots, \frac{k-1}{2}$ along with:
\begin{align}\label{eq:sympl_gf_const3}
\sum_{j=(k-1)\beta+1}^{k\beta}\left[\sum_{l=1}^{j}y_l^{(2j)}-2\sum_{l=1}^{j}y_l^{(2j-1)}+\right.&\left.\sum_{l=1}^{j-1}y_l^{(2j-2)}\right]\\
  &=\sum_{j=(k-1)\beta+1}^{k\beta}\left[\sum_{l=1}^{j}y_l^{(4k\beta-2j)}-2\sum_{l=1}^{j}y_l^{(4k\beta-2j+1)}+\sum_{l=1}^{j-1}y_l^{(4k\beta-2j+2)}\right].\nonumber
\end{align}
Observe that, for both $k$ odd and even there are a total of $k$ constraints. 
\end{enumerate}
\end{definition}

We claim that there is a natural bijection, essentially a relabelling of the coordinates, between $GT_{Sp}(N;k;\beta)$ and $\mathfrak{I}_{Sp}(N;k;\beta)$:
\begin{align}\label{eq:symbij}
\mathfrak{B}_{Sp}:GT_{Sp}(N;k;\beta) \longrightarrow \mathfrak{I}_{Sp}(N;k;\beta).
\end{align}
This can be seen as follows, and for additional clarity see Figure~\ref{fig:symplbij}. Let $(\lambda^{(i)})_{i=1}^{4k\beta} \in GT_{Sp}(N;k;\beta)$. Observe that, by the interlacing $\lambda^{(4k\beta-1)}\prec \langle N^{2k\beta}\rangle=\lambda^{(4k\beta)}$, we have a single free coordinate:
\begin{align*}
\lambda_{1}^{(4k\beta-1)},\dots,\lambda_{2k\beta-1}^{(4k\beta-1)}\equiv N,\\
0 \le \lambda_{2k\beta}^{(4k\beta-1)}\le N.
\end{align*} 
We thus relabel $y_1^{(4k\beta-1)}=\lambda_{2k\beta}^{(4k\beta-1)}$. Secondly, again due to the interlacing $\lambda^{(4k\beta-2)}\prec \lambda^{(4k\beta-1)}$, we have:
\begin{align*}
\lambda_{1}^{(4k\beta-2)},\dots,\lambda_{2k\beta-2}^{(4k\beta-2)}\equiv N
\end{align*}
and moreover,
\begin{align*}
y_1^{(4k\beta-1)}=\lambda_{2k\beta}^{(4k\beta-1)}\le \lambda_{2k\beta-1}^{(4k\beta-2)}\le N.
\end{align*}
We write $y_1^{(4k\beta-2)}=\lambda_{2k\beta-1}^{(4k\beta-2)}$. We continue relabelling in this fashion up to (and including) $\lambda^{(2k\beta+1)}$ (after which no coordinates are necessarily fixed to equal $N$) and finally, we put $(y^{(i)})_{i=1}^{2k\beta}\equiv(\lambda^{(i)})_{i=1}^{2k\beta}$. Clearly, the map $\mathfrak{B}_{Sp}$ described above is invertible. Thus, by making use of Proposition \ref{CombRepSymp1} we obtain the following.
\begin{proposition}\label{CombRepSymp2}
Let $k,\beta \in \mathbb{N}$. Then,
\begin{align*}
\mom_{Sp(2N)}(k,\beta)=\#\mathfrak{I}_{Sp}(N;k;\beta).
\end{align*}
\end{proposition}


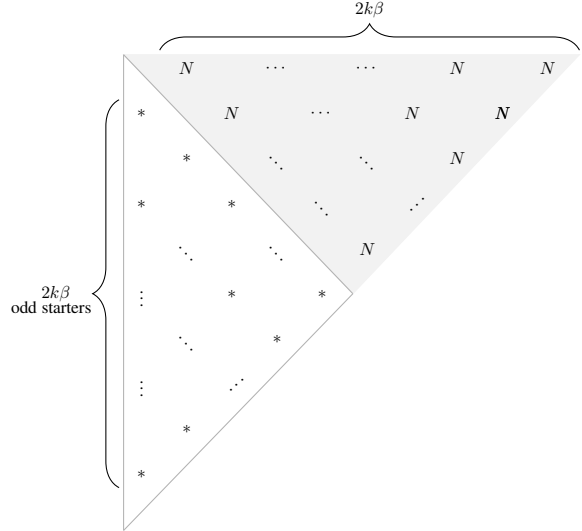
\begin{figure}[!htb]
\centering
\begin{tikzpicture}[scale=0.6, every node/.style={scale=0.6}]

\draw[decorate, decoration={brace, amplitude=10pt}] (0.4,9.4) -- (9.7,9.4) node [midway, yshift=.9cm] {$2k\beta$};

\path [draw=gray!10, fill=gray!10] (-0.4,9.3) -- (9.7,9.3) -- (4.68,4) -- cycle;
\node at (1,9) {$N$};
\node at (3,9) {$\cdots$};
\node at (5,9) {$\cdots$};
\node at (7,9) {$N$};
\node at (9,9) {$N$};
\node at (8,8) {$N$};
\node at (2,8) {$N$};
\node at (4,8) {$\cdots$};
\node at (6,8) {$N$};
\node at (8,8) {$N$};
\node at (7,7) {$N$};
\draw (6,6) node[rotate=80] {$\ddots$};
\draw (5,7) node[rotate=-10] {$\ddots$};
\draw (3,7) node[rotate=-10] {$\ddots$};
\draw (4,6) node[rotate=-10] {$\ddots$};
\draw (3,5) node[rotate=-10] {$\ddots$};
\node at (5,5) {$N$};

\draw[thin, color=gray!60] (-0.4,9.3) -- (-0.4,-1.25);
\draw[thin, color=gray!60] (4.68,4) -- (-0.4,-1.25);
\draw[thin, color=gray!60] (4.68,4) -- (-0.4,9.3);

\draw[decorate, decoration={brace, amplitude=10pt}]  (-0.5,-0.3) -- (-0.5,8.3) node [midway, xshift=-1.4cm] {$2k\beta$};
\node at (-2, 3.7) {odd-starters};
\node at (0,8) {$\ast$};
\node at (1,7) {$\ast$};
\node at (0,6) {$\ast$};
\node at (2,6) {$\ast$};
\node at (4,4) {$\ast$};
\draw (1,5) node[rotate=-10] {$\ddots$};
\node at (0,4) {$\vdots$};
\node at (0,2) {$\vdots$};
\draw (1,3) node[rotate=-10] {$\ddots$};
\draw (2,2) node[rotate=80] {$\ddots$};
\node at (2,4) {$\ast$};
\node at (3,3) {$\ast$};
\node at (1,1) {$\ast$};
\node at (0,0) {$\ast$};
\end{tikzpicture}

\caption{Figure depicting the fixed region of a $(4k\beta)$-symplectic Gelfand-Tsetlin pattern $P\in SP_{\langle N^{2k\beta}\rangle}$.  The shaded area represents the fixed region, whilst the unshaded region shows which elements have some freedom in the values that they can take.}\label{fig:symplfixed}
\end{figure}

\begin{figure}[!htb]
\centering
\begin{tikzpicture}[scale=0.6, every node/.style={scale=0.6}]

\path [draw=gray!10, fill=gray!10] (-0.2,9.3) -- (9.7,9.3) -- (4.88,4) -- cycle;
\node at (1,9) {$N$};
\node at (3,9) {$\cdots$};
\node at (5,9) {$\cdots$};
\node at (7,9) {$N$};
\node at (9,9) {$N$};
\node at (8,8) {$N$};
\node at (2,8) {$N$};
\node at (4,8) {$\cdots$};
\node at (6,8) {$N$};
\node at (8,8) {$N$};
\node at (7,7) {$N$};
\draw (6,6) node[rotate=80] {$\ddots$};
\draw (5,7) node[rotate=-10] {$\ddots$};
\draw (3,7) node[rotate=-10] {$\ddots$};
\draw (4,6) node[rotate=-10] {$\ddots$};
\draw (3,5) node[rotate=-10] {$\ddots$};
\node at (5,5) {$N$};

\node at (0,8) {$\lambda_{2k\beta}^{(4k\beta-1)}$};
\node at (1,7) {$\lambda_{2k\beta-1}^{(4k\beta-2)}$};
\node at (0,6) {$\lambda_{2k\beta-1}^{(4k\beta-3)}$};
\node at (2,6) {$\lambda_{2k\beta-2}^{(4k\beta-3)}$};
\draw (3,5) node[rotate=-10] {$\ddots$};
\node at (4,4) {$\lambda_{1}^{(2k\beta)}$};
\draw (1,5) node[rotate=-10] {$\ddots$};
\node at (0,4.5) {$\vdots$};
\node at (0,2.5) {$\vdots$};
\draw (1,3) node[rotate=-10] {$\ddots$};
\draw (2,2) node[rotate=80] {$\ddots$};
\node at (.7,3.8) {$\dots$};
\node at (2,4) {$\lambda_{2}^{(2k\beta)}$};
\node at (3,3) {$\lambda_{1}^{(2k\beta-1)}$};
\node at (1,1) {$\lambda_{1}^{(2)}$};
\node at (0,0) {$\lambda_{1}^{(1)}$};

\node at (13,9) {};
\node at (15,9) {};
\node at (17,9) {};
\node at (14,8) {};
\node at (16,8) {};
\node at (18,8) {};
\node at (18,6) {};
\draw (17,7) node[rotate=-10] {};
\draw (15,7) node[rotate=-10] {};
\draw (16,6) node[rotate=-10] {};
\draw (15,5) node[rotate=-10] {};
\node at (17,5) {};


\node at (12,8) {$y_1^{(4k\beta-1)}$};
\node at (13,7) {$y_1^{(4k\beta-2)}$};
\node at (12,6) {$y_2^{(4k\beta-3)}$};
\node at (14,6) {$y_1^{(4k\beta-3)}$};s
\node at (16,4) {$y_{1}^{(2k\beta)}$};
\draw (13,5) node[rotate=-10] {$\ddots$};
\draw (15,5) node[rotate=-10] {$\ddots$};
\node at (12,4.5) {$\vdots$};
\node at (12,2.5) {$\vdots$};
\draw (13,3) node[rotate=-10] {$\ddots$};
\draw (14,2) node[rotate=80] {$\ddots$};
\node at (12.7,3.8) {$\dots$};
\node at (14,4) {$y_{2}^{(2k\beta)}$};
\node at (15,3) {$y_{1}^{(2k\beta-1)}$};
\node at (13,1) {$y_{1}^{(2)}$};
\node at (12,0) {$y_{1}^{(1)}$};

\node at (9,5) {\huge{$\xrightarrow[]{\mathfrak{B}_{Sp}}$}};

\end{tikzpicture}

\caption{Representation of the relabelling of the coordinates given by the bijection $\mathfrak{B}_{Sp}:GT_{Sp}(N;k;\beta) \longrightarrow \mathfrak{I}_{Sp}(N;k;\beta)$ }\label{fig:symplbij}
\end{figure}
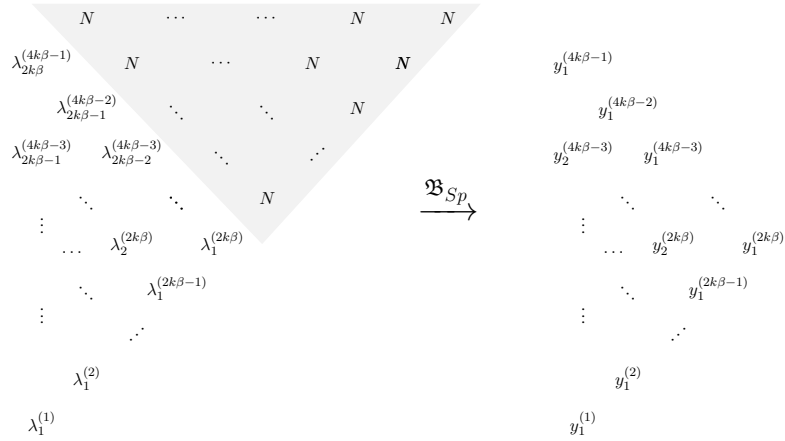


\begin{figure}[!htb]
\centering
\begin{tikzpicture}[scale=0.7, every node/.style={scale=0.7}]

\path [fill=gray!20, rounded corners] (0.5,7) -- (7.5,7) -- (4,3.5) -- cycle;

\path [pattern=north west lines, pattern color=gray!40] (0.52,6.88) -- (1.813,5.56) -- (0.52,5.56) -- cycle;
\path [pattern=north east lines, pattern color=gray!40] (0.52,5.76) -- (1.62,5.76) -- (2.9,4.45) -- (0.52,4.45) -- cycle;
\path [pattern=north west lines, pattern color=gray!40] (0.52,4.65) -- (2.71,4.65) -- (3.99,3.34) -- (0.52,3.34) -- cycle;
\path [pattern=north east lines, pattern color=gray!40] (0.52,2.23) -- (0.52,3.54) -- (3.99,3.54) -- (2.71,2.23) -- cycle;
\path [pattern=north west lines, pattern color=gray!40] (0.52,2.43) -- (2.9,2.43) -- (1.62,1.12) -- (0.52,1.12) -- cycle;
\path [pattern=north east lines, pattern color=gray!40] (0.52,0) -- (1.81,1.32) -- (0.52,1.32) -- cycle;

\draw[rotate around={-45:(3.82,3.56)},fill=white, color=white] (3.65,3.48) rectangle (4,3.63); 
\draw[rotate around={45:(3.82,3.27)},fill=white, color=white] (3.65,3.2) rectangle (4,3.35); 

\draw[|-|] (-0.8,0) -- (-0.8,7);
\draw[decorate, decoration={brace, amplitude=4pt}]  (0.35,5.76) -- (0.35,6.88) node [midway, xshift=-.7cm] {\tiny{$4\beta-1$}};
\draw[decorate, decoration={brace, amplitude=4pt}]  (0.35,4.65) -- (0.35,5.66) node [midway, xshift=-.5cm] {\tiny{$4\beta$}};
\draw[decorate, decoration={brace, amplitude=4pt}]  (0.35,3.54) -- (0.35,4.55) node [midway, xshift=-.5cm] {\tiny{$4\beta$}};
\draw[decorate, decoration={brace, amplitude=4pt}]  (0.35,2.43) -- (0.35,3.44) node [midway, xshift=-.5cm] {\tiny{$4\beta$}};
\draw[decorate, decoration={brace, amplitude=4pt}]  (0.35,1.32) -- (0.35,2.33) node [midway, xshift=-.5cm] {\tiny{$4\beta$}};
\draw[decorate, decoration={brace, amplitude=4pt}]  (0.35,0) -- (0.35,1.22) node [midway, xshift=-.5cm] {\tiny{$4\beta$}};
\draw[decorate, decoration={brace, amplitude=4pt}]  (0.5,7.05) -- (7.5,7.05) node [midway, yshift=.45cm] {\small{Fixed}};
\node at (-1.2,3.45) {\small{$24\beta$}};

\node[mark size=1pt, color=black] at (0.56, 6.65) {$\bigcdot$};
\node[mark size=1pt, color=black] at (0.56, 5.62) {$\bigcdot$};
\node[mark size=1pt, color=black] at (0.56, 4.5) {$\bigcdot$};
\node[mark size=1pt, color=black] at (0.56,  3.4) {$\bigcdot$};
\node[mark size=1pt, color=black] at (0.56, 2.3) {$\bigcdot$};
\node[mark size=1pt, color=black] at (0.56, 1.15) {$\bigcdot$};

\end{tikzpicture}

\caption{Visual representations of how the index sets $\mathcal{S}_{(k,\beta)}^{Sp}$, and hence the diagram given by $\mathcal{V}_{(k,\beta)}^{Sp}$, for general integer $\beta$ and $k=6$ are constructed. A pair $(i, j)$ in index set $\mathcal{S}_{(k,\beta)}^{Sp}$ represents any non-fixed element $i$ in row $j$ of the continuous pattern $\mathcal{V}_{(k,\beta)}^{Sp}$ above, except for the elements depicted by $\bullet$. These are not included in $\mathcal{S}_{(k,\beta)}^{Sp}$, since these are chosen to be fixed by the linear equations.  The overlap in the pattern shows the $5$ rows $x^{(4\beta)},\dots,x^{(20\beta)}$ where the constraints overlap.}\label{fig:indexset}
\end{figure}
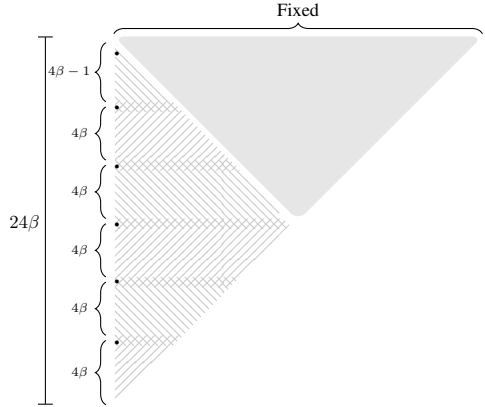

\subsection{Asymptotics and the leading order coefficient}\label{sec:sympl_asympt}

To conclude the proof, we require some final definitions and notation, which will also be useful for the orthogonal case in Section~\ref{sec:orthog}. We consider the continuous Weyl chamber:
\begin{align*}
\mathsf{W}_N=\{ x=(x_1,\dots,x_N)\in \mathbb{R}^N:x_1\ge \dots \ge x_N\}
\end{align*}
and also let $\mathsf{W}_N^+=\mathsf{W}_N\cap \mathbb{R}^N_+$. We say that $y \in \mathsf{W}_N$ and $x \in \mathsf{W}_{N+1}$ interlace if exactly the inequalities~\eqref{interlacing_eqs} (from the discrete setting) are satisfied and we also write $y\prec x$ (similarly for $y \in \mathsf{W}_N$ and $x\in \mathsf{W}_N$). The definitions of continuous half-patterns and continuous symplectic and orthogonal Gelfand-Tsetlin patterns are completely analogous to the discrete setting (we simply replace $\mathsf{S}_i$ by $\mathsf{W}_i$).

We consider the following index set, which encodes a subset of the elements in the patterns in $\mathfrak{I}_{Sp}(N;k;\beta)$ resulting from applying the relabelling,

\begin{align*}
\mathcal{S}^{Sp}_{(k,\beta)} \coloneqq \bigg\{ (m,n):\ &1 \le m \le \bigg\lfloor \frac{n+1}{2} \bigg\rfloor\text{ and }1 \le n \le 2k\beta; \\
&\text{or }1 \le m \le \bigg\lfloor \frac{4k\beta-n+1}{2}\bigg\rfloor\text{ and }2k\beta+1\le n < 4k\beta-1;\\
 &n\neq 4\beta, 8\beta, \dots, 4(k-1)\beta \bigg\}\\
 \cup \bigg\{ (m,4n\beta):& \ 1\le m \le 2n\beta-1 \text{ and }1\le n\le \left\lfloor \tfrac{k}{2} \right\rfloor;\\
 &\text{or }1\le m \le 2(k-n)\beta-1 \text{ and }\left\lfloor\tfrac{k}{2}\right\rfloor+1\le n < k \}\bigg\}.
\end{align*}
Thus, the pair $(m,n)$ appears in $\mathcal{S}^{Sp}_{(k,\beta)}$ if and only if $y_m^{(n)}\in \mathfrak{I}_{Sp}(N;k;\beta)$, except for some particular choices of pairs $(m,n)$, which we remove.  The $k$ missing pairs are precisely the encodings of $y_{2\beta}^{(4\beta)}$, $y_{4\beta}^{(8\beta)}$, $\dots$, $y_{4\beta}^{(4(k-2)\beta)}$, $y_{2\beta}^{(4(k-1)\beta}$, and $y_{1}^{(4k\beta-1)}$; see Figure~\ref{fig:indexset} for a visual representation. 

Observe that $\mathcal{S}^{Sp}_{(k,\beta)}$ has exactly $k\beta (2k\beta+1)-k$ elements. Now define 
\begin{align}\label{sympl_set_free}
\mathcal{V}_{(k,\beta)}^{Sp}\coloneqq \{ x_{m}^{(n)}\in \mathbb{R} : (m,n)\in \mathcal{S}^{Sp}_{(k,\beta)}, 0\le x_m^{(n)} \le 1\}\subset \mathbb{R}^{k\beta(2k\beta+1)-k},
\end{align}
alongside elements defined as follows, 
\begin{align}
x_{\frac{n}{2}}^{(n)}&\quad\text{for }n= 4\beta, 8\beta, \dots, 4\left\lfloor\tfrac{k}{2}\right\rfloor\beta,\label{extra_coord_sympl1}\\
x_{\frac{4k\beta-n}{2}}^{(n)}&\quad\text{for }n=4(\left\lfloor\tfrac{k}{2}\right\rfloor+1)\beta, \dots, 4(k-1)\beta, \label{extra_coord_sympl2}\\
x_{1}^{(4k\beta-1)},&\quad\label{extra_coord_sympl3}
\end{align}
which are determined by the linear equations~\eqref{eq:sympl_gf_const1}--\eqref{eq:sympl_gf_const3} (we simply solve for the relevant term) so that:
\begin{itemize}
	\item $0 \le x_{m}^{(n)} \le 1$, for all $x_m^{(n)}$ described by \eqref{sympl_set_free}--\eqref{extra_coord_sympl3},
	\item $x^{(n)},x^{(4k\beta-n)} \in \mathsf{W}_{\lfloor\frac{n+1}{2}\rfloor}^+$, for all $n=1,\dots,2k\beta$,
	\item both $(x^{(n)})_{n=1}^{2k\beta}$ and $(x^{(4k\beta-n)})_{n=1}^{2k\beta}$ form continuous $(2k\beta)$-symplectic Gelfand-Tsetlin patterns.
\end{itemize}
We call the index set corresponding to the `determined' elements 
\[\mathcal{T}^{Sp}_{(k,\beta)}\coloneqq \{(m,n) : y_m^{(n)}\in \mathfrak{I}_{Sp}(N;k;\beta)\} \backslash \mathcal{S}^{Sp}_{(k,\beta)}.\]

Observe that, $\mathcal{V}_{(k,\beta)}^{Sp}$ is convex as an intersection of hyperplanes. Moreover, $\mathcal{V}_{(k,\beta)}^{Sp}$ is contained in the cube $[0,1]^{k\beta(2k\beta+1)-k}$ and hence in a closed ball of radius $\sqrt{k\beta(2k\beta+1)-k}$.

\begin{proof}[Proof of Theorem \ref{MainTheoremSymplectic}]
The proof of the aspect of the theorem pertaining to the polynomial structure of the moments of moments was given in Proposition \ref{polystrucsympl}. For the leading order coefficient term we observe that:
\begin{align*}
\#\mathfrak{I}_{Sp}(N;k;\beta)=\# \left(\mathbb{Z}^{k\beta(2k\beta+1)-k}\cap\left(N\mathcal{V}_{(k,\beta)}^{Sp}\right)\right),
\end{align*}
where for a set $\mathcal{A}$, we write $N \mathcal{A}=\{Nx: x \in \mathcal{A} \}$ for is its dilate by a factor of $N$. Thus, from Proposition \ref{CombRepSymp2} and Theorem \ref{LatticePointCountTheorem} with $\mathcal{S}=N\mathcal{V}_{(k,\beta)}^{Sp}$, we obtain:
\begin{align*}
\mom_{Sp(2N)}(k,\beta)=\#\mathfrak{I}_{Sp}(N;k;\beta)&=\# \left(\mathbb{Z}^{k\beta(2k\beta+1)-k}\cap\left(N\mathcal{V}_{(k,\beta)}^{Sp}\right)\right)\\
&=\mathsf{vol}\left(N\mathcal{V}_{(k,\beta)}^{Sp}\right)+O_{k,\beta}\left(N^{k\beta(2k\beta+1)-k-1}\right).
\end{align*}
Since,
\begin{align*}
\mathsf{vol}\left(N\mathcal{V}_{(k,\beta)}^{Sp}\right)=N^{k\beta(2k\beta+1)-k}\mathsf{vol}\left(\mathcal{V}_{(k,\beta)}^{Sp}\right)
\end{align*}
we have $\mathfrak{c}_{Sp}(k,\beta)=\mathsf{vol}\left(\mathcal{V}_{(k,\beta)}^{Sp}\right)$. It then suffices to prove that $\mathsf{vol}\left(\mathcal{V}_{(k,\beta)}^{Sp}\right)>0$ which is the content of Lemma \ref{PositivityOfVolumeSymplectic} below.
\end{proof}

Proving the strict positivity of the constant $\mathfrak{c}_{Sp}(k,\beta)$ is important, because otherwise we simply have a bound for $\mom_{Sp(2N)}(k,\beta)$. This task is also one of the more complicated parts of this paper. A crucial role is played by a number of figures which elucidate the argument.

\begin{lemma}\label{PositivityOfVolumeSymplectic}
Let $k,\beta \in \mathbb{N}$. Then,
\begin{align}
\mathfrak{c}_{Sp}(k,\beta)=\mathsf{vol}\left(\mathcal{V}_{(k,\beta)}^{Sp}\right)>0.
\end{align} 
\end{lemma}

\begin{proof}
We consider the following subset $\tilde{\mathcal{V}}_{(k,\beta)}^{Sp}\subset \mathcal{V}_{(k,\beta)}^{Sp}$ defined as for $\mathcal{V}_{(k,\beta)}^{Sp}$, but additional we require both that $0 <x_m^{(n)}<1$ and the interlacing is strict:
\begin{align*}
x_m^{(n+1)}>x_m^{(n)}>x_{m+1}^{(n+1)},
\end{align*}
the above holding also for $x_m^{(n)}$ for $(m,n)\in\mathcal{T}^{Sp}_{(k,\beta)}$ as given in~\eqref{extra_coord_sympl1}--\eqref{extra_coord_sympl3}. Now, we claim that if there exists at least one element in $\tilde{\mathcal{V}}_{(k,\beta)}^{Sp}$ then $\mathsf{vol}\left(\tilde{\mathcal{V}}_{(k,\beta)}^{Sp}\right)>0$ since $\tilde{\mathcal{V}}_{(k,\beta)}^{Sp}$ contains a small cube around this element (this clearly implies the statement of the lemma). This can easily be seen as follows. Take a continuous pattern $P = (z_m^{(n)})_{(m,n)\in \mathcal{S}^{Sp}_{(k,\beta)} } \in \tilde{\mathcal{V}}_{(k,\beta)}^{Sp}$ and let $\mathsf{d}$ be the minimal distance between any two elements $z_m^{(n)}$ of $P$, or between $z_m^{(n)}$ and $0$ or $1$ (including those $z_m^{(n)}$ corresponding to the points described in~\cref{extra_coord_sympl1,extra_coord_sympl2,extra_coord_sympl3}). We observe that if we change each of the coordinates $(z_m^{(n)})_{(m,n)\in \mathcal{S}^{Sp}_{(k,\beta)} }$ by at most some positive $\epsilon$, then there exists some constant $C_{k,\beta}$ such that the extra values
given by  $z_m^{(n)}$ for $(m,n)\in\mathcal{T}^{Sp}_{(k,\beta)}$ change by at most $C_{k,\beta}\times \epsilon$. Thus, if $\epsilon=\epsilon(\mathsf{d})$ is small enough we get that $(z_m^{(n)})_{(m,n)\in \mathcal{S}^{Sp}_{(k,\beta)} }+[-\epsilon,\epsilon]^{k\beta(2k\beta+1)-k} \subset \tilde{\mathcal{V}}_{(k,\beta)}^{Sp}$.

It then suffices to exhibit such an element. We observe that the constraints described in \eqref{eq:sympl_gf_const1}--\eqref{eq:sympl_gf_const3} essentially fall in to four distinct categories, hereafter types $1$, $2$, $3$, and $4$.  These can be visualised as in Figures~\ref{fig:constr1}, ~\ref{fig:constr2}, ~\ref{fig:constr3}, and ~\ref{fig:constr4}.  In each diagram, the shaded triangular region shows the part of the pattern $P\in GT_{Sp}(N;k;\beta)$ which was fixed to be $N$, and the numbers shown to the left of the pattern are the `row coefficient'.  One can reconstruct the particular constraint described in each figure by first multiplying each row sum by its row coefficient, and the summing the resulting expressions for the top half of the pattern, and equating it with the sum for the bottom half of the pattern (the `symmetry line' is given by the row with row coefficient $0$).   For example, Figure~\ref{fig:constr1} shows the following constraint, ($k=1, \beta=3$ in \eqref{eq:sympl_gf_const3}), 

\begin{align*}
\sum_{j=1}^{3}\left[\sum_{l=1}^{j}y_l^{(2j)}-2\sum_{l=1}^{j}y_l^{(2j-1)}+\sum_{l=1}^{j-1}y_l^{(2j-2)}\right]&=\sum_{j=1}^{3}\left[\sum_{l=1}^{j}y_l^{(12-2j)}-2\sum_{l=1}^{j}y_l^{(13-2j)}+\sum_{l=1}^{j-1}y_l^{(14-2j)}\right]\\
  \intertext{or, equivalently,}
2\sum_{j=1}^{5}(-1)^jr^{(j)}&=2\sum_{j=7}^{11}(-1)^jr^{(j)},
\end{align*}
where $r^{(j)}$ is the sum of the elements in row $j$. 

We will first show that it is possible to exhibit an element with strict interlacing and positive distances from $0$ and $1$ for each of the four types of constraints. We will then argue that these constructions are compatible and yield an element of $\tilde{\mathcal{V}}_{(k,\beta)}^{Sp}$; this fact is not entirely trivial since two consecutive constraints (e.g. $i=1,2$ in \eqref{eq:sympl_gf_const1}) overlap in a single row, see Figures~\ref{fig:constr5} and ~\ref{fig:constr6}, and clearly interlacing still plays a role.

The first two types of constraints, types $1$ and $2$ are shown in Figures~\ref{fig:constr1} and ~\ref{fig:constr2}.  Type $1$ only occurs for $k=1$ and Figure~\ref{fig:constr1} shows an example for $k=1$ and $\beta=3$.  In this case, only \eqref{eq:sympl_gf_const3} is relevant.  The row sum for the ($2k\beta$)th row appears on both sides of ~\eqref{eq:sympl_gf_const3}, and so this contribution is cancelled out.  All the remaining row sums have a coefficient of either $+2$ or $-2$ in~\eqref{eq:sympl_gf_const3}, and precisely which coefficient corresponds to which row can be seen on the left in Figure~\ref{fig:constr1}.  Similarly, type $2$ is the generalisation of type $1$ but for $k>1$, odd.  For these larger values of odd $k$, the shape of the constraint changes from triangular to pentagonal, but always occurs in the centre portion of the overall pattern. Figure~\ref{fig:constr2} shows the type $2$ for $k=3$ and $\beta=2$.  For both said constraints, it is easy to exhibit such an element by symmetry: simply pick the lower half-pattern to have strict interlacing and coordinates a positive distance away from $0$ and $1$ and reflect in the symmetry line (c.f. the row with factor $0$ in either figure).

Constraints of types $3$ and $4$ are shown in Figures~\ref{fig:constr3} and ~\ref{fig:constr4}.  Type $3$ occurs for $k\geq 2$ and corresponds to \cref{eq:sympl_gf_const1,eq:sympl_gf_const2} for $i=1$ - henceforth we say that a `lower' type $3$ pattern comes from setting $i=1$ in \eqref{eq:sympl_gf_const1}; whereas an `upper' type $3$ pattern is the analogous object using \eqref{eq:sympl_gf_const2}.  The shape of type $3$ is always triangular and covers the lowermost and uppermost portion of the overall pattern (c.f. the top and bottom patterned triangles in Figure~\ref{fig:indexset}). Figure~\ref{fig:constr3} shows type $3$ for $k=2$, $\beta=2$, and in particular the lower version, corresponding to $i=1$ in~\eqref{eq:sympl_gf_const1}.  Note now that all rows have coefficients that are either $\pm2$, except for the top (resp. for the upper version, bottom) row which gets a coefficient of $1$.  Type $4$ occurs for $k\geq 4$ and represents $i>1$ in  \cref{eq:sympl_gf_const1,eq:sympl_gf_const2}; the terms `lower' and `upper' are used just as for type $3$.  Type $4$ constraints are trapezoidal, and an example of the lower type is drawn in Figure~\ref{fig:constr4} for $k=4, \beta=2$.  Here (as for the general case) the row coefficients are once again symmetrical around the `overlap' row. For type $3$ and type $4$ constraints, exhibiting an element is more complicated than type $1$ and $2$, and we proceed as follows. 

In case of a constraint of type $3$, we split the configuration as in Figure~\ref{fig:constr7}. This results in a type $1$ constraint and a new constraint, hereafter referred to as type $5$.  In Figure~\ref{fig:constr7}, the top diagram gives an example of this splitting for a general form of a lower type $3$, and the particular form of the resulting type $5$ constraint is shown in the bottom diagram.  For the constraint of type $1$ resulting from the splitting, we will again use symmetry. However, the constraint of type $5$ requires a separate argument. 

Take $\epsilon>0$ to be very small according to $k$ and $\beta$. We pick the lower half-pattern of constraint type $1$, see Figure~\ref{fig:constr7}, so that the distances between any two nearest coordinates, and between the closest coordinate to $0$ (and respectively $1$), is strictly positive and at most $\epsilon$. We then use reflection through the middle row (the row with $0$ as its row coefficient) for the upper half-pattern. We then proceed to the constraint of type $5$. We again pick the coordinates, except the largest one (see circled element in Figure~\ref{fig:constr7}) to be at a strictly positive distance of at most $\epsilon$ to its neighbour coordinates, and to the edge of the upper half-pattern of the constraint of type $1$. Then, the total sum corresponding to constraint type $5$ excluding the largest coordinate, which we have yet to pick, is negative and at most $c_{k,\beta} \times \epsilon$ in absolute value, for some constant $c_{k,\beta}$ depending only on $k$ and $\beta$. We can then pick the largest coordinate so that this weighted sum over all coordinates is zero as long as $c_{k,\beta} \times \epsilon<1$. 

In order to deal with a constraint of type $4$ we split it into a constraint of type $2$ and type $5$, see Figure~\ref{fig:constr8}. There, the general `lower' type $4$ constraint is shown, along with the method of splitting.   One may use exactly the same method described above for type $3$ constraints.

Finally, we need to argue that using the procedures above is compatible with putting constraints together.  For example, type $3$ and type $4$ constraints overlap, see~\Cref{fig:constr5,fig:constr9}, and two type $4$ constraints also may overlap, see~\Cref{fig:constr6,fig:constr10}. With a mixture of type $3$ and type $4$ (the case for a mixture of two type $4$s is analogous), if we use the algorithm above to satisfy the constraint of type $3$, then the interlacing forces the coordinates at the edges of the next constraint of type $4$ to be `large', of the order of $c_{k,\beta}\times \epsilon$ for the constant $c_{k,\beta}$ described above. This then forces the largest coordinate of the constraint of type $5$ coming from the splitting of the constraint of type $4$ to be $\tilde{c}_{k,\beta}\times \epsilon$ for some (possibly much) larger constant $\tilde{c}_{k,\beta}$. However, we note that this does not present any real problems since we only need to apply this procedure a finite number of times and thus as long as we pick $\epsilon$ small enough so that $c_{k,\beta}^* \times \epsilon <1$ for some finite and fixed constant $c_{k,\beta}^*$, the result is as claimed.
\end{proof}



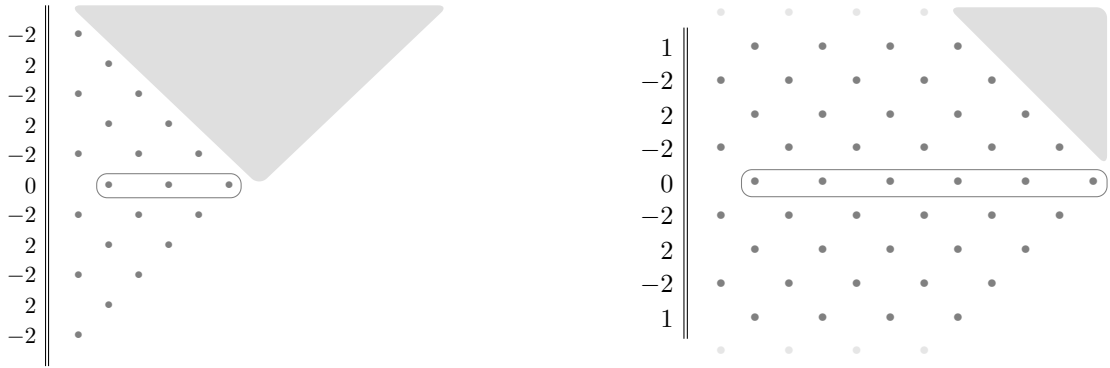
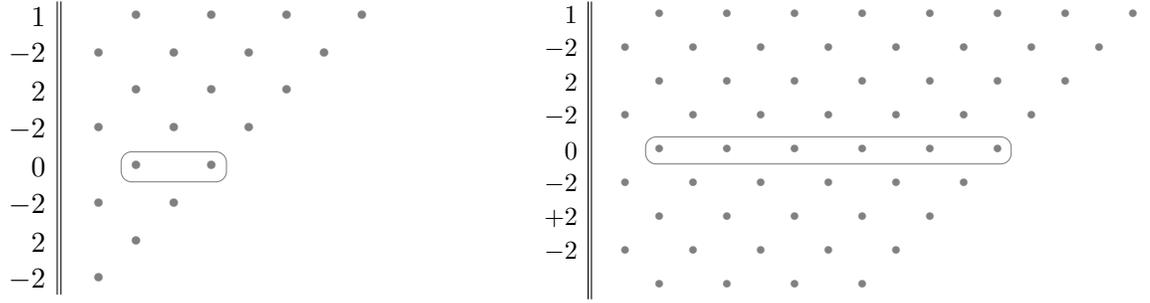
\begin{figure}[!htb]
\centering
\begin{subfigure}[t]{.48\textwidth}
  \centering
  \begin{tikzpicture}[scale=0.8, every node/.style={scale=0.8}]
\path [fill=gray!25, rounded corners] (0.35,3) -- (3.5,0) -- (6.65,3) -- cycle;

\draw[thin] (0,-3) -- (0,3);
\draw[thin] (-0.05,-3) -- (-.05,3);
\node at (-.3,0) {$0$};
\node at (-.45,0.5) {$-2$};
\node at (-.3,1) {$2$};
\node at (-.45,1.5) {$-2$};
\node at (-.3,2) {$2$};
\node at (-.45,2.5) {$-2$};
\node at (-.45,-0.5) {$-2$};
\node at (-.3,-1) {$2$};
\node at (-.45,-1.5) {$-2$};
\node at (-.3,-2) {$2$};
\node at (-.45,-2.5) {$-2$};

\node[mark size=2pt, color=gray] at (1,0) {$\bbigcdot$};
\node[mark size=2pt, color=gray] at (2,0) {$\bbigcdot$};
\node[mark size=2pt, color=gray] at (3,0) {$\bbigcdot$};
\node[mark size=2pt, color=gray] at (0.5,0.5) {$\bbigcdot$};
\node[mark size=2pt, color=gray] at (1.5,0.5) {$\bbigcdot$};
\node[mark size=2pt, color=gray] at (2.5,0.5) {$\bbigcdot$};
\node[mark size=2pt, color=gray] at (0.5,-0.5) {$\bbigcdot$};
\node[mark size=2pt, color=gray] at (1.5,-0.5) {$\bbigcdot$};
\node[mark size=2pt, color=gray] at (2.5,-0.5) {$\bbigcdot$};
\node[mark size=2pt, color=gray] at (1,1) {$\bbigcdot$};
\node[mark size=2pt, color=gray] at (2,1) {$\bbigcdot$};
\node[mark size=2pt, color=gray] at (1,-1) {$\bbigcdot$};
\node[mark size=2pt, color=gray] at (2,-1) {$\bbigcdot$};
\node[mark size=2pt, color=gray] at (0.5,1.5) {$\bbigcdot$};
\node[mark size=2pt, color=gray] at (1.5,1.5) {$\bbigcdot$};
\node[mark size=2pt, color=gray] at (0.5,-1.5) {$\bbigcdot$};
\node[mark size=2pt, color=gray] at (1.5,-1.5) {$\bbigcdot$};
\node[mark size=2pt, color=gray] at (1,2) {$\bbigcdot$};
\node[mark size=2pt, color=gray] at (0.5,2.5) {$\bbigcdot$};
\node[mark size=2pt, color=gray] at (1,-2) {$\bbigcdot$};
\node[mark size=2pt, color=gray] at (0.5,-2.5) {$\bbigcdot$};

\draw[rounded corners, color=gray] (0.8,-0.2) rectangle (3.2,0.2) {};

\end{tikzpicture}
\caption{Example of constraint type $1$.  This occurs exclusively for $k=1$, and is drawn for $k=1$, $\beta=3$.  The circled coordinates are those which feature in the `overlap' of the constraint~\eqref{eq:sympl_gf_const3}. The grey shaded area shows which elements are fixed to be $N$. The numbers on the left show the coefficient that appears against any given row sum in~\eqref{eq:sympl_gf_const3}.}\label{fig:constr1}
\end{subfigure}\hfill
\begin{subfigure}[t]{.48\textwidth}
  \centering
\begin{tikzpicture}[scale=0.9, every node/.style={scale=0.9}]

\path [rounded corners, fill=gray!25] (3.85,2.6) -- (6.2,0.25) -- (6.2,2.6) -- cycle;

\draw[thin] (0,-2.3) -- (0,2.3);
\draw[thin] (-0.05,-2.3) -- (-.05,2.3);
\node at (-.3,0) {$0$};
\node at (-.45,0.5) {$-2$};
\node at (-.3,1) {$2$};
\node at (-.45,1.5) {$-2$};
\node at (-.3,2) {$1$};
\node at (-.45,-0.5) {$-2$};
\node at (-.3,-1) {$2$};
\node at (-.45,-1.5) {$-2$};
\node at (-.3,-2) {$1$};

\node[mark size=2pt, color=gray] at (1,0) {$\bbigcdot$};
\node[mark size=2pt, color=gray] at (2,0) {$\bbigcdot$};
\node[mark size=2pt, color=gray] at (3,0) {$\bbigcdot$};
\node[mark size=2pt, color=gray] at (4,0) {$\bbigcdot$};
\node[mark size=2pt, color=gray] at (5,0) {$\bbigcdot$};
\node[mark size=2pt, color=gray] at (6,0) {$\bbigcdot$};
\node[mark size=2pt, color=gray] at (0.5,0.5) {$\bbigcdot$};
\node[mark size=2pt, color=gray] at (1.5,0.5) {$\bbigcdot$};
\node[mark size=2pt, color=gray] at (2.5,0.5) {$\bbigcdot$};
\node[mark size=2pt, color=gray] at (3.5,0.5) {$\bbigcdot$};
\node[mark size=2pt, color=gray] at (4.5,0.5) {$\bbigcdot$};
\node[mark size=2pt, color=gray] at (5.5,0.5) {$\bbigcdot$};
\node[mark size=2pt, color=gray] at (0.5,-0.5) {$\bbigcdot$};
\node[mark size=2pt, color=gray] at (1.5,-0.5) {$\bbigcdot$};
\node[mark size=2pt, color=gray] at (2.5,-0.5) {$\bbigcdot$};
\node[mark size=2pt, color=gray] at (3.5,-0.5) {$\bbigcdot$};
\node[mark size=2pt, color=gray] at (4.5,-0.5) {$\bbigcdot$};
\node[mark size=2pt, color=gray] at (5.5,-0.5) {$\bbigcdot$};
\node[mark size=2pt, color=gray] at (1,1) {$\bbigcdot$};
\node[mark size=2pt, color=gray] at (2,1) {$\bbigcdot$};
\node[mark size=2pt, color=gray] at (3,1) {$\bbigcdot$};
\node[mark size=2pt, color=gray] at (4,1) {$\bbigcdot$};
\node[mark size=2pt, color=gray] at (5,1) {$\bbigcdot$};
\node[mark size=2pt, color=gray] at (1,-1) {$\bbigcdot$};
\node[mark size=2pt, color=gray] at (2,-1) {$\bbigcdot$};
\node[mark size=2pt, color=gray] at (3,-1) {$\bbigcdot$};
\node[mark size=2pt, color=gray] at (4,-1) {$\bbigcdot$};
\node[mark size=2pt, color=gray] at (5,-1) {$\bbigcdot$};
\node[mark size=2pt, color=gray] at (0.5,1.5) {$\bbigcdot$};
\node[mark size=2pt, color=gray] at (1.5,1.5) {$\bbigcdot$};
\node[mark size=2pt, color=gray] at (2.5,1.5) {$\bbigcdot$};
\node[mark size=2pt, color=gray] at (3.5,1.5) {$\bbigcdot$};
\node[mark size=2pt, color=gray] at (4.5,1.5) {$\bbigcdot$};
\node[mark size=2pt, color=gray] at (0.5,-1.5) {$\bbigcdot$};
\node[mark size=2pt, color=gray] at (1.5,-1.5) {$\bbigcdot$};
\node[mark size=2pt, color=gray] at (2.5,-1.5) {$\bbigcdot$};
\node[mark size=2pt, color=gray] at (3.5,-1.5) {$\bbigcdot$};
\node[mark size=2pt, color=gray] at (4.5,-1.5) {$\bbigcdot$};
\node[mark size=2pt, color=gray] at (1,2) {$\bbigcdot$};
\node[mark size=2pt, color=gray] at (2,2) {$\bbigcdot$};
\node[mark size=2pt, color=gray] at (3,2) {$\bbigcdot$};
\node[mark size=2pt, color=gray] at (4,2) {$\bbigcdot$};
\node[mark size=2pt, color=gray] at (1,-2) {$\bbigcdot$};
\node[mark size=2pt, color=gray] at (2,-2) {$\bbigcdot$};
\node[mark size=2pt, color=gray] at (3,-2) {$\bbigcdot$};
\node[mark size=2pt, color=gray] at (4,-2) {$\bbigcdot$};

\node[mark size=2pt, color=gray!20] at (0.5,-2.5) {$\bbigcdot$};
\node[mark size=2pt, color=gray!20] at (1.5,-2.5) {$\bbigcdot$};
\node[mark size=2pt, color=gray!20] at (2.5,-2.5) {$\bbigcdot$};
\node[mark size=2pt, color=gray!20] at (3.5,-2.5) {$\bbigcdot$};
\node[mark size=2pt, color=gray!20] at (0.5,2.5) {$\bbigcdot$};
\node[mark size=2pt, color=gray!20] at (1.5,2.5) {$\bbigcdot$};
\node[mark size=2pt, color=gray!20] at (2.5,2.5) {$\bbigcdot$};
\node[mark size=2pt, color=gray!20] at (3.5,2.5) {$\bbigcdot$};

\draw[rounded corners, color=gray] (0.8,-0.2) rectangle (6.2,0.2) {};

\end{tikzpicture}

\caption{Example of constraint type $2$.  This occurs for $k>1$, $k$ odd, and is drawn for $k=3$, $\beta=2$.  The circled coordinates are those which feature in the `overlap' of constraint~\eqref{eq:sympl_gf_const3} (i.e. those in row $2k\beta$). The grey shaded area shows the lower part of the section which is fixed to be $N$, and the number on the left show the coefficient that appears against any given row sum in~\eqref{eq:sympl_gf_const3}. }\label{fig:constr2}

\end{subfigure}
\caption{Figures showing constraints of type $1$ and $2$ for the symplectic case.}
\label{fig:type1and2}
\end{figure}


\begin{figure}[!htb]
\centering
\begin{subfigure}[t]{0.48\textwidth}
\centering
\begin{tikzpicture}
\draw[thin] (0,-0.2) -- (0,3.7);
\draw[thin] (-0.05,-0.2) -- (-.05,3.7);
\node at (-.45,0) {$-2$};
\node at (-.3,0.5) {$2$};
\node at (-.45,1) {$-2$};
\node at (-.3,1.5) {$0$};
\node at (-.45,2) {$-2$};
\node at (-.3,2.5) {$2$};
\node at (-.45,3) {$-2$};
\node at (-.3,3.5) {$1$};

\node[mark size=2pt, color=gray] at (0.5,0) {$\bbigcdot$};
\node[mark size=2pt, color=gray] at (1,0.5) {$\bbigcdot$};
\node[mark size=2pt, color=gray] at (0.5,1) {$\bbigcdot$};
\node[mark size=2pt, color=gray] at (1.5,1) {$\bbigcdot$};
\node[mark size=2pt, color=gray] at (1,1.5) {$\bbigcdot$};
\node[mark size=2pt, color=gray] at (2,1.5) {$\bbigcdot$};
\node[mark size=2pt, color=gray] at (0.5,2) {$\bbigcdot$};
\node[mark size=2pt, color=gray] at (1.5,2) {$\bbigcdot$};
\node[mark size=2pt, color=gray] at (2.5,2) {$\bbigcdot$};
\node[mark size=2pt, color=gray] at (1,2.5) {$\bbigcdot$};
\node[mark size=2pt, color=gray] at (2,2.5) {$\bbigcdot$};
\node[mark size=2pt, color=gray] at (3,2.5) {$\bbigcdot$};
\node[mark size=2pt, color=gray] at (0.5,3) {$\bbigcdot$};
\node[mark size=2pt, color=gray] at (1.5,3) {$\bbigcdot$};
\node[mark size=2pt, color=gray] at (2.5,3) {$\bbigcdot$};
\node[mark size=2pt, color=gray] at (3.5,3) {$\bbigcdot$};
\node[mark size=2pt, color=gray] at (1,3.5) {$\bbigcdot$};
\node[mark size=2pt, color=gray] at (2,3.5) {$\bbigcdot$};
\node[mark size=2pt, color=gray] at (3,3.5) {$\bbigcdot$};
\node[mark size=2pt, color=gray] at (4,3.5) {$\bbigcdot$};

\draw[rounded corners, color=gray] (0.8,1.3) rectangle (2.2,1.7) {};

\end{tikzpicture}

\caption{Example of constraint type $3$.  This occurs for $k\ge2$, and is partly drawn for $k=2$, $\beta=2$.  The figure depicts the first constraint  (i.e. $i=1$ in \eqref{eq:sympl_gf_const1}) and the boxed elements are those which appear in the `overlap' of said constraint. Note that by reflecting this diagram in the $x$-plane, one gets a figure for the last constraint, i.e. $i=1$ in \eqref{eq:sympl_gf_const2}. The numbers on the left are the coefficients that appear against the relevant row in \eqref{eq:sympl_gf_const1}, with $i=1$. }\label{fig:constr3}
\end{subfigure}\hfill
\begin{subfigure}[t]{0.48\textwidth}
\centering
\begin{tikzpicture}[scale=0.9, every node/.style={scale=0.9}]
\draw[thin] (0,-0.7) -- (0,3.7);
\draw[thin] (-0.05,-0.7) -- (-.05,3.7);
\node at (-.3,-0.5) {};
\node at (-.45,0) {$-2$};
\node at (-.45,0.5) {$+2$};
\node at (-.45,1) {$-2$};
\node at (-.3,1.5) {$0$};
\node at (-.45,2) {$-2$};
\node at (-.3,2.5) {$2$};
\node at (-.45,3) {$-2$};
\node at (-.3,3.5) {$1$};
\node[mark size=2pt, color=gray] at (1,-0.5) {$\bbigcdot$};
\node[mark size=2pt, color=gray] at (2,-0.5) {$\bbigcdot$};
\node[mark size=2pt, color=gray] at (3,-0.5) {$\bbigcdot$};
\node[mark size=2pt, color=gray] at (4,-0.5) {$\bbigcdot$};
\node[mark size=2pt, color=gray] at (0.5,0) {$\bbigcdot$};
\node[mark size=2pt, color=gray] at (1.5,0) {$\bbigcdot$};
\node[mark size=2pt, color=gray] at (2.5,0) {$\bbigcdot$};
\node[mark size=2pt, color=gray] at (3.5,0) {$\bbigcdot$};
\node[mark size=2pt, color=gray] at (4.5,0) {$\bbigcdot$};
\node[mark size=2pt, color=gray] at (1,0.5) {$\bbigcdot$};
\node[mark size=2pt, color=gray] at (2,0.5) {$\bbigcdot$};
\node[mark size=2pt, color=gray] at (3,0.5) {$\bbigcdot$};
\node[mark size=2pt, color=gray] at (4,0.5) {$\bbigcdot$};
\node[mark size=2pt, color=gray] at (5,0.5) {$\bbigcdot$};
\node[mark size=2pt, color=gray] at (0.5,1) {$\bbigcdot$};
\node[mark size=2pt, color=gray] at (1.5,1) {$\bbigcdot$};
\node[mark size=2pt, color=gray] at (2.5,1) {$\bbigcdot$};
\node[mark size=2pt, color=gray] at (3.5,1) {$\bbigcdot$};
\node[mark size=2pt, color=gray] at (4.5,1) {$\bbigcdot$};
\node[mark size=2pt, color=gray] at (5.5,1) {$\bbigcdot$};
\node[mark size=2pt, color=gray] at (1,1.5) {$\bbigcdot$};
\node[mark size=2pt, color=gray] at (2,1.5) {$\bbigcdot$};
\node[mark size=2pt, color=gray] at (3,1.5) {$\bbigcdot$};
\node[mark size=2pt, color=gray] at (4,1.5) {$\bbigcdot$};
\node[mark size=2pt, color=gray] at (5,1.5) {$\bbigcdot$};
\node[mark size=2pt, color=gray] at (6,1.5) {$\bbigcdot$};
\node[mark size=2pt, color=gray] at (0.5,2) {$\bbigcdot$};
\node[mark size=2pt, color=gray] at (1.5,2) {$\bbigcdot$};
\node[mark size=2pt, color=gray] at (2.5,2) {$\bbigcdot$};
\node[mark size=2pt, color=gray] at (3.5,2) {$\bbigcdot$};
\node[mark size=2pt, color=gray] at (4.5,2) {$\bbigcdot$};
\node[mark size=2pt, color=gray] at (5.5,2) {$\bbigcdot$};
\node[mark size=2pt, color=gray] at (6.5,2) {$\bbigcdot$};
\node[mark size=2pt, color=gray] at (1,2.5) {$\bbigcdot$};
\node[mark size=2pt, color=gray] at (2,2.5) {$\bbigcdot$};
\node[mark size=2pt, color=gray] at (3,2.5) {$\bbigcdot$};
\node[mark size=2pt, color=gray] at (4,2.5) {$\bbigcdot$};
\node[mark size=2pt, color=gray] at (5,2.5) {$\bbigcdot$};
\node[mark size=2pt, color=gray] at (6,2.5) {$\bbigcdot$};
\node[mark size=2pt, color=gray] at (7,2.5) {$\bbigcdot$};
\node[mark size=2pt, color=gray] at (0.5,3) {$\bbigcdot$};
\node[mark size=2pt, color=gray] at (1.5,3) {$\bbigcdot$};
\node[mark size=2pt, color=gray] at (2.5,3) {$\bbigcdot$};
\node[mark size=2pt, color=gray] at (3.5,3) {$\bbigcdot$};
\node[mark size=2pt, color=gray] at (4.5,3) {$\bbigcdot$};
\node[mark size=2pt, color=gray] at (5.5,3) {$\bbigcdot$};
\node[mark size=2pt, color=gray] at (6.5,3) {$\bbigcdot$};
\node[mark size=2pt, color=gray] at (7.5,3) {$\bbigcdot$};
\node[mark size=2pt, color=gray] at (1,3.5) {$\bbigcdot$};
\node[mark size=2pt, color=gray] at (2,3.5) {$\bbigcdot$};
\node[mark size=2pt, color=gray] at (3,3.5) {$\bbigcdot$};
\node[mark size=2pt, color=gray] at (4,3.5) {$\bbigcdot$};
\node[mark size=2pt, color=gray] at (5,3.5) {$\bbigcdot$};
\node[mark size=2pt, color=gray] at (6,3.5) {$\bbigcdot$};
\node[mark size=2pt, color=gray] at (7,3.5) {$\bbigcdot$};
\node[mark size=2pt, color=gray] at (8,3.5) {$\bbigcdot$};

\draw[rounded corners, color=gray] (0.8,1.3) rectangle (6.2,1.7) {};

\end{tikzpicture}

\caption{Example of constraint type $4$.  This occurs for $k\geq4$ and is drawn for $k=4, \beta=2$ and depicts the (lower) constraint for $i=2$ in \eqref{eq:sympl_gf_const1}.  The boxed elements are those which feature in the `overlap' of the described constraint, and the numbers on the left give the coefficient of a given row sum in \eqref{eq:sympl_gf_const1}.  Note that the shape and row coefficients of the upper constraint can be seen by reflecting the diagram in the $x$-plane. }\label{fig:constr4}
\end{subfigure}
\caption{Figures showing constraints of type 3 and 4 for the symplectic case.}
\end{figure}


\begin{figure}[!htb]
\centering
\begin{tikzpicture}

\node at (11.5,9) {Splitting of Type 3};
\path [fill=gray!10,rounded corners] (7.3,7.5) -- (15.7,7.5) -- (11.5,3.5) -- cycle;

\draw[thin] (7,-0.2) -- (7,7);
\draw[thin] (6.95,-0.2) -- (6.95,7);
\node at (6.55,0) {$-2$};
\node at (6.7,0.5) {$2$};
\node at (6.55,1) {$-2$};
\node at (6.7,1.8) {$\vdots$};
\node at (6.7,3) {$\vdots$};
\node at (6.7,3.5) {$0$};
\node at (6.55,4) {$-2$};
\node at (6.7,4.5) {$2$};
\node at (6.7,5.3) {$\vdots$};
\node at (6.7,6.1) {$\vdots$};
\node at (6.55,6.7) {$-2$};

\draw[thin] (16,3.8) -- (16,7.5);
\draw[thin] (15.95,3.8) -- (15.95,7.5);
\node at (16.55,4.5) {$2$};
\node at (16.4,4) {$-2$};
\node at (16.55,5.3) {$\vdots$};
\node at (16.55,6.1) {$\vdots$};
\node at (16.4,6.7) {$-2$};
\node at (16.55,7.2) {$1$};

\node[mark size=2pt, color=gray] at (7.5,0) {$\bbigcdot$};
\node[mark size=2pt, color=gray] at (8,0.5) {$\bbigcdot$};
\node[mark size=2pt, color=gray] at (7.5,1) {$\bbigcdot$};
\node[mark size=2pt, color=gray] at (8.5,1) {$\bbigcdot$};
\node[mark size=2pt, color=gray] at (7.5,1.8) {\vdots};
\draw (9,1.8) node[rotate=85,color=gray] {$\ddots$};
\node[mark size=2pt, color=gray] at (7.5,3) {\vdots};
\draw (10.2,2.85) node[rotate=85,color=gray] {$\ddots$};
\node[mark size=2pt, color=gray] at (8,3.5) {$\bbigcdot$};
\node[color=gray] at (9.5,3.5) {$\cdots$};
\node[mark size=2pt, color=gray] at (11,3.5) {$\bbigcdot$};
\draw[decorate, decoration={brace, amplitude=10pt}]  (11,3.4) --  (8,3.4) node [midway, yshift=-18pt] {\small{$\beta$}};
\node[mark size=2pt, color=gray] at (7.5,3.95) {$\bbigcdot$};
\node[mark size=2pt, color=gray] at (8.5,3.95) {$\bbigcdot$};
\node[color=gray] at (9.5,4) {$\cdots$};
\node[mark size=2pt, color=gray] at (10.5,3.95) {$\bbigcdot$};
\node[mark size=2pt, color=gray] at (11.5,3.95) {$\bbigcdot$};
\draw[decorate, decoration={brace, amplitude=10pt}]  (7.5,4.1) -- (10.5,4.1) node [midway, yshift=18pt] {\small{$\beta$}};
\node[mark size=2pt, color=gray] at (8,4.5) {$\bbigcdot$};
\node[color=gray] at (9,4.5) {$\cdots$};
\node[mark size=2pt, color=gray] at (10,4.5) {$\bbigcdot$};
\node[mark size=2pt, color=gray] at (11,4.5) {$\bbigcdot$};
\node[mark size=2pt, color=gray] at (12,4.5) {$\bbigcdot$};
\node[mark size=2pt, color=gray] at (7.5,5.3) {\vdots};
\draw (9.6,5.15) node[rotate=-6,color=gray] {$\ddots$};
\draw (10.3,5.2) node[rotate=-6,color=gray] {$\ddots$};
\draw (8.6,6.15) node[rotate=-6,color=gray] {$\ddots$};
\draw (9.3,6.2) node[rotate=-6,color=gray] {$\ddots$};
\draw (12.7,5.2) node[rotate=80,color=gray] {$\ddots$};
\draw (13.7,6.2) node[rotate=80,color=gray] {$\ddots$};
\node[mark size=2pt, color=gray] at (7.5,6.1) {\vdots};
\node[mark size=2pt, color=gray] at (7.5,6.7) {$\bbigcdot$};
\node[mark size=2pt, color=gray] at (8.5,6.7) {$\bbigcdot$};
\node[mark size=2pt, color=gray] at (8,7.17) {$\bbigcdot$};
\node[mark size=2pt, color=gray] at (9,7.17) {$\bbigcdot$};
\node[color=gray] at (10.5,7.2) {$\cdots$};
\node[color=gray] at (12.5,7.2) {$\cdots$};
\node[color=gray] at (9.5,6.7) {$\cdots$};
\node[color=gray] at (11.5,6.7) {$\cdots$};
\node[color=gray] at (13.5,6.7) {$\cdots$};
\node[mark size=2pt, color=gray] at (14.5,6.7) {$\bbigcdot$};
\node[mark size=2pt, color=gray] at (14,7.17) {$\bbigcdot$};
\node[mark size=2pt, color=gray] at (15,7.17) {$\bbigcdot$};
\draw[decorate, decoration={brace, amplitude=10pt}]  (8,7.5) -- (15,7.5) node [midway, yshift=18pt] {\small{$2\beta$}};
\draw[gray!10!black] (15,7.2) circle (5pt);

\draw[decorate, decoration={brace, amplitude=10pt}]  (6.2,-0.2) -- (6.2,7) node [midway, xshift=-25pt] {\small{Type 1}};
\draw[decorate, decoration={brace, amplitude=10pt}] (17,7.5) -- (17,3.8) node [midway, xshift=25pt] {\small{Type 5}};

\node at (11.5,-0.75) {Type 5};
\draw[decorate, decoration={brace, amplitude=10pt}]  (8,-2) -- (15,-2) node [midway, yshift=18pt] {\small{$2\beta$}};
\path [fill=gray!10,rounded corners] (7.3,-2) -- (15.7,-2) -- (11.5,-6) -- cycle;

\draw[thin, <->] (7,-6) -- (7,-2);
\node at (6.65,-4) {$2\beta$};
\node at (7.4,-2.4) {\small{$1$}};
\node at (7.8,-2.9) {\small{$-2$}};
\draw (8.8,-3.75) node[rotate=-6] {$\ddots$};
\node at (9.6,-4.6) {\small{$-2$}};
\node at (10.2,-5.1) {\small{$2$}};
\node at (10.6,-5.6) {\small{$-2$}};

\node at (15.8,-2.4) {\small{$y^{(2\beta)}$}};
\node at (15.5,-2.9) {\small{$y^{(2\beta-1)}$}};
\draw (13.9,-3.9) node[rotate=80] {$\ddots$};
\node at (13.4,-4.6) {\small{$y^{(3)}$}};
\node at (12.9,-5.1) {\small{$y^{(2)}$}};
\node at (12.3,-5.6) {\small{$y^{(1)}$}};

\node[mark size=2pt, color=gray] at (8.5,-2.85) {$\bbigcdot$};
\node[mark size=2pt, color=gray] at (8,-2.33) {$\bbigcdot$};
\node[mark size=2pt, color=gray] at (9,-2.33) {$\bbigcdot$};
\node[color=gray] at (10.5,-2.3) {$\cdots$};
\node[color=gray] at (12.5,-2.3) {$\cdots$};
\node[color=gray] at (9.5,-2.8) {$\cdots$};
\node[color=gray] at (11.5,-2.8) {$\cdots$};
\node[color=gray] at (13.5,-2.8) {$\cdots$};
\node[mark size=2pt, color=gray] at (14.5,-2.85) {$\bbigcdot$};
\node[mark size=2pt, color=gray] at (14,-2.33) {$\bbigcdot$};
\node[mark size=2pt, color=gray] at (15,-2.33) {$\bbigcdot$};
\draw (13,-4.1) node[rotate=80,color=gray] {$\ddots$};
\draw (13.75,-3.4) node[rotate=80,color=gray] {$\ddots$};
\draw (9.15,-3.3) node[rotate=-6,color=gray] {$\ddots$};
\draw (9.9,-4) node[rotate=-6,color=gray] {$\ddots$};
\node[mark size=2pt, color=gray] at (12.5,-4.55) {$\bbigcdot$};
\node[mark size=2pt, color=gray] at (11.5,-4.55) {$\bbigcdot$};
\node[mark size=2pt, color=gray] at (10.5,-4.55) {$\bbigcdot$};
\node[mark size=2pt, color=gray] at (12,-5.05) {$\bbigcdot$};
\node[mark size=2pt, color=gray] at (11,-5.05) {$\bbigcdot$};
\node[mark size=2pt, color=gray] at (11.5,-5.55) {$\bbigcdot$};
\draw[gray!10!black] (15,-2.3) circle (5pt);
\end{tikzpicture}

\caption{Figures giving the construction of a type $5$ constraint, which comes from splitting a type $3$ constraint (see Figure~\ref{fig:constr3}).  This occurs for $k\ge2$, and the version for a lower type $3$ constraint (i.e. $i=1$ in \eqref{eq:sympl_gf_const1}) is drawn in the upper figure to show the situation for general $\beta$, and $k\geq 2$.  The type $3$ constraint is split in to one of type $1$ (the unshaded region) and one of a new type, type $5$ (the shaded region). The bottom figure shows explicitly the constraint of type $5$, which forms a Gelfand-Tsetlin pattern $(y^{(i)})_{i=1}^{2\beta}$, where $y^{(i)}\in \mathsf{W}^+_i$ and $y^{(i)}\prec y^{(i+1)}$. In both diagrams, the circled top right element is the largest, and the numbers on either side show the row sum weightings for $i=1$ in \eqref{eq:sympl_gf_const1}.  The equivalent form for the upper version (i.e. $i=1$ in \eqref{eq:sympl_gf_const2}) can be seen by reflecting the top diagram in the $x$-plane.}\label{fig:constr7}
\end{figure}
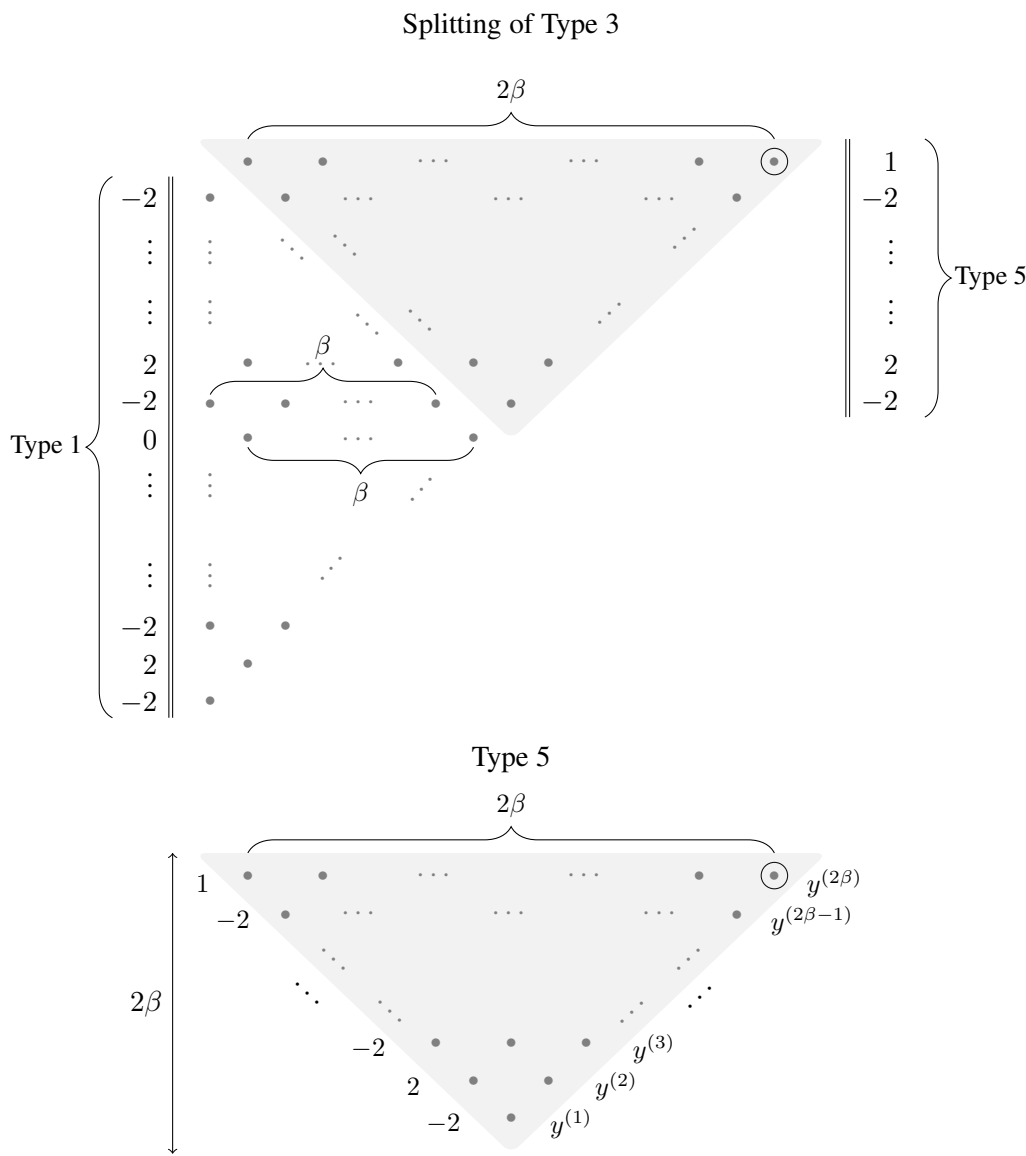


\begin{figure}[!htb]
\centering
\begin{tikzpicture}[scale=0.75, every node/.style={scale=0.75}]

\path [fill=gray!10,rounded corners] (7.3,7.5) -- (14.7,7.5) -- (11,4) -- cycle;

\draw[thin] (2,0.3) -- (2,7.5);
\draw[thin] (1.95,0.3) -- (1.95,7.5);
\node at (1.7,0.5) {$1$};
\node at (1.55,1) {$-2$};
\node at (1.7,2) {$\vdots$};
\node at (1.7,2.9) {$\vdots$};
\node at (1.55,3.5) {$-2$};
\node at (1.7,4) {$0$};
\node at (1.55,4.5) {$-2$};
\node at (1.7,5.3) {$\vdots$};
\node at (1.7,6.1) {$\vdots$};
\node at (1.55,6.7) {$-2$};
\node at (1.7,7.2) {$1$};

\draw[thin] (16,4.2) -- (16,7.5);
\draw[thin] (15.95,4.2) -- (15.95,7.5);
\node at (16.4,4.5) {$-2$};
\node at (16.55,5.3) {$\vdots$};
\node at (16.55,6.1) {$\vdots$};
\node at (16.4,6.7) {$-2$};
\node at (16.55,7.2) {$1$};

\node[mark size=2pt, color=gray] at (3.5,0.5) {$\bbigcdot$};
\node[color=gray] at (4.5,0.5) {$\cdots$};
\node[mark size=2pt, color=gray] at (6,0.5) {$\bbigcdot$};
\node[mark size=2pt, color=gray] at (7,0.5) {$\bbigcdot$};
\node[mark size=2pt, color=gray] at (3,1) {$\bbigcdot$};
\node[color=gray] at (4,1) {$\cdots$};
\node[color=gray] at (6.5,1) {$\cdots$};
\node[mark size=2pt, color=gray] at (7.5,1) {$\bbigcdot$};
\node[mark size=2pt, color=gray] at (3,2.5) {\vdots};
\draw (8.2,1.7) node[rotate=80,color=gray] {$\ddots$};
\draw (9.4,2.8) node[rotate=80,color=gray] {$\ddots$};
\node[mark size=2pt, color=gray] at (3,3.5) {$\bbigcdot$};
\node[color=gray] at (4,3.5) {$\cdots$};
\node[color=gray] at (8,3.5) {$\cdots$};
\node[mark size=2pt, color=gray] at (9,3.5) {$\bbigcdot$};
\node[mark size=2pt, color=gray] at (10,3.5) {$\bbigcdot$};
\draw (3.5,3) node[rotate=-6,color=gray] {$\ddots$};
\node[mark size=2pt, color=gray] at (3.5,4) {$\bbigcdot$};
\node[color=gray] at (4.4,4) {$\cdots$};
\node[color=gray] at (8.5,4) {$\cdots$};
\node[mark size=2pt, color=gray] at (9.5,4) {$\bbigcdot$};
\node[mark size=2pt, color=gray] at (10.5,4) {$\bbigcdot$};
\node[mark size=2pt, color=gray] at (3,4.5) {$\bbigcdot$};
\node[color=gray] at (4,4.5) {$\cdots$};
\node[color=gray] at (8,4.5) {$\cdots$};
\node[mark size=2pt, color=gray] at (9,4.5) {$\bbigcdot$};
\node[mark size=2pt, color=gray] at (10,4.5) {$\bbigcdot$};
\node[mark size=2pt, color=gray] at (11,4.5) {$\bbigcdot$};
\node[mark size=2pt, color=gray] at (3,5.8) {\vdots};
\draw (9.4,5.15) node[rotate=-6,color=gray] {$\ddots$};
\draw (10.3,5.2) node[rotate=-6,color=gray] {$\ddots$};
\draw (8.2,6.15) node[rotate=-6,color=gray] {$\ddots$};
\draw (9.3,6.2) node[rotate=-6,color=gray] {$\ddots$};
\draw (11.7,5.2) node[rotate=80,color=gray] {$\ddots$};
\draw (12.7,6.2) node[rotate=80,color=gray] {$\ddots$};
\node[mark size=2pt, color=gray] at (3,6.7) {$\bbigcdot$};
\node[color=gray] at (4, 6.7) {$\cdots$};
\node[color=gray] at (6.3, 6.7) {$\cdots$};
\node[mark size=2pt, color=gray] at (7.5,6.7) {$\bbigcdot$};
\node[mark size=2pt, color=gray] at (8.5,6.7) {$\bbigcdot$};
\node[mark size=2pt, color=gray] at (3.5,7.17) {$\bbigcdot$};
\draw (3.5,6.2) node[rotate=-6,color=gray] {$\ddots$};
\node[color=gray] at (4.5,7.2) {$\cdots$};
\node[mark size=2pt, color=gray] at (6,7.17) {$\bbigcdot$};
\node[mark size=2pt, color=gray] at (7,7.17) {$\bbigcdot$};
\node[mark size=2pt, color=gray] at (8,7.17) {$\bbigcdot$};
\node[mark size=2pt, color=gray] at (9,7.17) {$\bbigcdot$};
\node[color=gray] at (10,7.2) {$\cdots$};
\node[color=gray] at (12,7.2) {$\cdots$};
\node[color=gray] at (9.5,6.7) {$\cdots$};
\node[color=gray] at (11,6.7) {$\cdots$};
\node[color=gray] at (12.5,6.7) {$\cdots$};
\node[mark size=2pt, color=gray] at (13.5,6.7) {$\bbigcdot$};
\node[mark size=2pt, color=gray] at (13,7.17) {$\bbigcdot$};
\node[mark size=2pt, color=gray] at (14,7.17) {$\bbigcdot$};

\node at (8.2,0.5) {$\lambda^{(2(2i-2)\beta)}$};
\node at (9,1) {${\lambda^{(2(2i-2)\beta+1)}}$};
\node at (11.4,3.5) {${\lambda^{(2(2i-1)\beta-1)}}$};
\node at (12,4) {${\lambda^{(2(2i-1)\beta)}}$};
\node at (12.8,4.5) {${\lambda^{(2(2i-1)\beta+1)}}$};
\node at (14.7,6.7) {${\lambda^{(4i\beta-1)}}$};
\node at (15.2,7.2) {${\lambda^{(4i\beta)}}$};

\draw[gray!10!black] (14,7.2) circle (5pt);
\draw[decorate, decoration={brace, amplitude=10pt}]  (1.2,0.3) -- (1.2,7.5) node [midway, xshift=-30pt] {\small{Type 2}};
\draw[decorate, decoration={brace, amplitude=10pt}] (16.8,7.5) -- (16.8,4.2) node [midway, xshift=30pt] {\small{Type 5}};

\end{tikzpicture}

\caption{Figure showing splitting a type $4$ constraint (see Figure~\ref{fig:constr4}) in to a type $2$ and type $5$.  This occurs for $k\ge4$, and the lower constraint for some $1<i\leq\left\lfloor \frac{k}{2}\right\rfloor$ in \eqref{eq:sympl_gf_const1} is drawn in the top figure for general $k\geq 4, \beta$, involving rows $\lambda^{(n)}$ for $n=2(2i-2)\beta,\dots, 4i\beta$.  The type $4$ constraint is split in to one of type $2$ (the unshaded region) and one of type $5$ (the shaded region), see Figure~\ref{fig:constr7}. The circled top right element is the largest, and the numbers on the far left and the far right give the row sum weightings as appearing in \eqref{eq:sympl_gf_const1}.  The equivalent form for the upper version (i.e. $1<i\leq\left\lfloor \frac{k}{2}\right\rfloor$ in \eqref{eq:sympl_gf_const2}) can be seen by reflecting the diagram in the $x$-plane.}\label{fig:constr8}
\end{figure}
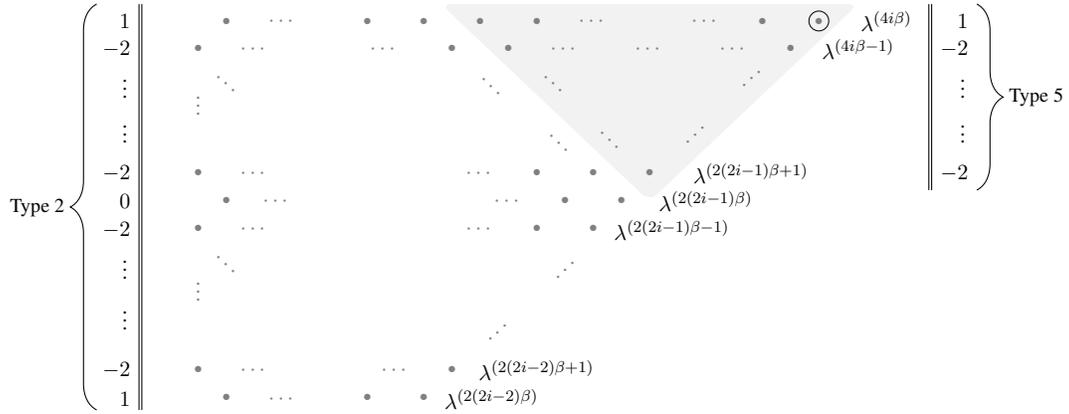


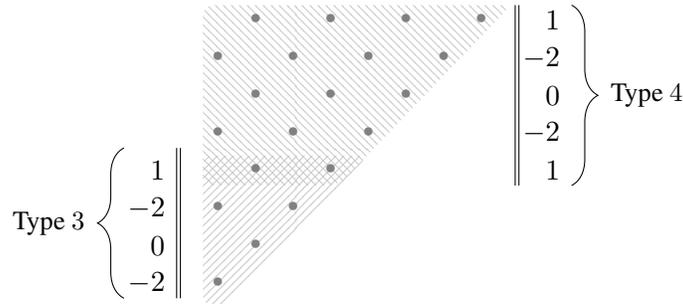
\begin{figure}[!htb]
\begin{tikzpicture}

\path [pattern=north west lines, pattern color=gray!40,rounded corners] (0.3,1.3) -- (0.3,3.7) -- (4.4,3.7) -- (2.1,1.3) -- cycle;
\path [pattern=north east lines, pattern color=gray!40, rounded corners] (0.3, -0.3) -- (0.3,1.7) -- (2.5,1.7) --  (0.6, -0.3)   -- cycle;

\draw[thin] (0,-0.2) -- (0,1.8);
\draw[thin] (-0.05,-0.2) -- (-.05,1.8);
\draw[thin] (4.45, 1.3) -- (4.45,3.7);
\draw[thin] (4.5,1.3) -- (4.5,3.7);
\node at (-.45,0) {$-2$};
\node at (-.3,0.5) {$0$};
\node at (-.45,1) {$-2$};
\node at (-.3,1.5) {$1$};
\node at (4.95,1.5) {$1$};
\node at (4.8,2) {$-2$};
\node at (4.95,2.5) {$0$};
\node at (4.8,3) {$-2$};
\node at (4.95,3.5) {$1$};

\draw[decorate, decoration={brace, amplitude=10pt}]  (5.2,3.7) -- (5.2,1.3) node [midway, xshift=1cm] {\small{Type $4$}};
\draw[decorate, decoration={brace, amplitude=10pt}]  (-0.75,-0.2) -- (-0.75,1.8)  node [midway, xshift=-1cm] {\small{Type $3$}};

\node[mark size=2pt, color=gray] at (0.5,0) {$\bbigcdot$};
\node[mark size=2pt, color=gray] at (1,0.5) {$\bbigcdot$};
\node[mark size=2pt, color=gray] at (0.5,1) {$\bbigcdot$};
\node[mark size=2pt, color=gray] at (1.5,1) {$\bbigcdot$};
\node[mark size=2pt, color=gray] at (1,1.5) {$\bbigcdot$};
\node[mark size=2pt, color=gray] at (2,1.5) {$\bbigcdot$};
\node[mark size=2pt, color=gray] at (0.5,2) {$\bbigcdot$};
\node[mark size=2pt, color=gray] at (1.5,2) {$\bbigcdot$};
\node[mark size=2pt, color=gray] at (2.5,2) {$\bbigcdot$};
\node[mark size=2pt, color=gray] at (1,2.5) {$\bbigcdot$};
\node[mark size=2pt, color=gray] at (2,2.5) {$\bbigcdot$};
\node[mark size=2pt, color=gray] at (3,2.5) {$\bbigcdot$};
\node[mark size=2pt, color=gray] at (0.5,3) {$\bbigcdot$};
\node[mark size=2pt, color=gray] at (1.5,3) {$\bbigcdot$};
\node[mark size=2pt, color=gray] at (2.5,3) {$\bbigcdot$};
\node[mark size=2pt, color=gray] at (3.5,3) {$\bbigcdot$};
\node[mark size=2pt, color=gray] at (1,3.5) {$\bbigcdot$};
\node[mark size=2pt, color=gray] at (2,3.5) {$\bbigcdot$};
\node[mark size=2pt, color=gray] at (3,3.5) {$\bbigcdot$};
\node[mark size=2pt, color=gray] at (4,3.5) {$\bbigcdot$};


\end{tikzpicture}
\caption{Example of a mixture of type 3 and type 4.  This example shows $k=4$, $\beta=1$, and the interplay between $i=1$ and $i=2$ in \eqref{eq:sympl_gf_const1} is demonstrated through the overlap between the two patterns.  The corresponding diagram for $i=1$ and $i=2$ in \eqref{eq:sympl_gf_const2} is simply the reflection of this diagram in the $x$-plane.}\label{fig:constr5}
\end{figure}


\begin{figure}[!htb]
\centering
\begin{tikzpicture}

\path [pattern=north west lines, pattern color=gray!40,rounded corners] (0,0) -- (0,3) -- (3,3) -- cycle;
\path [pattern=north east lines, pattern color=gray!40, rounded corners] (0,2.75) -- (0,5) -- (5,5) -- (2.75,2.75)-- cycle;

\draw[] (0,3) -- (1.5,1.5);
\draw[dashed] (0,1.5) -- (1.5,1.5);
\draw[] (3.85,3.85) -- (2.7,5);
\draw[dashed] (0,3.85) -- (3.85,3.85);
\node[mark size=1.5pt, color=gray] at (2.6,2.83) {$\bbigcdot$};
\node[mark size=1.5pt, color=gray] at (4.6,4.83) {$\bbigcdot$};
\draw[gray!10!black] (4.6,4.85) circle (3pt);
\draw[gray!10!black] (2.6,2.85) circle (3pt);

\draw[decorate, decoration={brace, amplitude=10pt}]  (-0.2,2.7) -- (-0.2,5) node [midway, xshift=-1.2cm] {\small{Type $4$}};
\draw[decorate, decoration={brace, amplitude=10pt}]  (-0.4,0) -- (-0.4,3) node [midway, xshift=-1cm] {\small{Type $3$}};
\draw[decorate, decoration={brace, amplitude=3pt}]  (3.1,3.05) -- (3.1,2.65) node [midway, xshift=0.55cm] {\tiny{overlap}};
\draw[->] (2.78,2.9) -- (3.68,3.8);
\draw[->] (2.78,4.8) -- (3.68,3.9);

\node at (1.9,1.5) {\tiny{mirror}};
\node at (4.25,3.85) {\tiny{mirror}};

\end{tikzpicture}
\caption{Example of combining a split type $3$ and a split type $4$.  The dashed horizontal lines represent the lines of reflection, and the solid diagonal lines show where the splitting of the respective types occurs.  The circled elements are the largest element for each section, and the arrows show the location of elements that, due to the interlacing, are forced the be `large', and also direction of growth.   }\label{fig:constr9}
\end{figure}
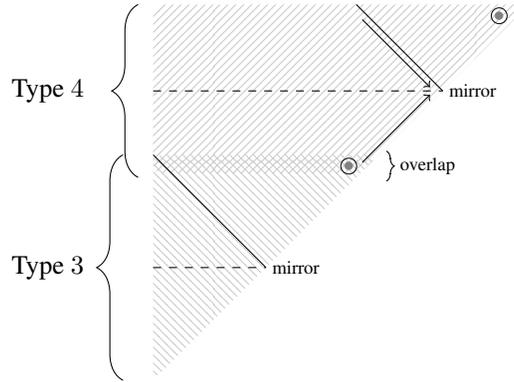

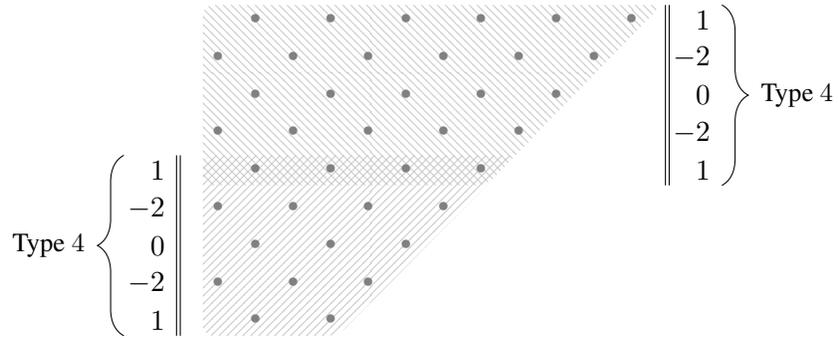
\begin{figure}[!htb]
\begin{tikzpicture}

\path [pattern=north west lines, pattern color=gray!40,rounded corners] (0.3,1.3) -- (0.3,3.7) -- (6.4,3.7) -- (4.07,1.3) -- cycle;
\path [pattern=north east lines, pattern color=gray!40, rounded corners] (0.3, -0.7) -- (0.3,1.7) -- (4.47,1.7) --  (2.1, -0.7)   -- cycle;

\draw[thin] (0,-0.7) -- (0,1.7);
\draw[thin] (-0.05,-0.7) -- (-.05,1.7);
\draw[thin] (6.5,1.3) -- (6.5,3.7);
\draw[thin] (6.45,1.3) -- (6.45,3.7);
\node at (-.3,-0.5) {$1$};
\node at (-.45,0) {$-2$};
\node at (-.3,0.5) {$0$};
\node at (-.45,1) {$-2$};
\node at (-0.3,1.5) {$1$};
\node at (6.95,1.5) {$1$};
\node at (6.8,2) {$-2$};
\node at (6.95,2.5) {$0$};
\node at (6.8,3) {$-2$};
\node at (6.95,3.5) {$1$};

\draw[decorate, decoration={brace, amplitude=10pt}]  (7.2,3.7) -- (7.2,1.3) node [midway, xshift=1cm] {\small{Type $4$}};
\draw[decorate, decoration={brace, amplitude=10pt}]  (-0.75,-0.7) -- (-0.75,1.7)  node [midway, xshift=-1cm] {\small{Type $4$}};

\node[mark size=2pt, color=gray] at (1,-0.5) {$\bbigcdot$};
\node[mark size=2pt, color=gray] at (2,-0.5) {$\bbigcdot$};
\node[mark size=2pt, color=gray] at (0.5,0) {$\bbigcdot$};
\node[mark size=2pt, color=gray] at (1.5,0) {$\bbigcdot$};
\node[mark size=2pt, color=gray] at (2.5,0) {$\bbigcdot$};
\node[mark size=2pt, color=gray] at (1,0.5) {$\bbigcdot$};
\node[mark size=2pt, color=gray] at (2,0.5) {$\bbigcdot$};
\node[mark size=2pt, color=gray] at (3,0.5) {$\bbigcdot$};
\node[mark size=2pt, color=gray] at (0.5,1) {$\bbigcdot$};
\node[mark size=2pt, color=gray] at (1.5,1) {$\bbigcdot$};
\node[mark size=2pt, color=gray] at (2.5,1) {$\bbigcdot$};
\node[mark size=2pt, color=gray] at (3.5,1) {$\bbigcdot$};
\node[mark size=2pt, color=gray] at (1,1.5) {$\bbigcdot$};
\node[mark size=2pt, color=gray] at (2,1.5) {$\bbigcdot$};
\node[mark size=2pt, color=gray] at (3,1.5) {$\bbigcdot$};
\node[mark size=2pt, color=gray] at (4,1.5) {$\bbigcdot$};
\node[mark size=2pt, color=gray] at (0.5,2) {$\bbigcdot$};
\node[mark size=2pt, color=gray] at (1.5,2) {$\bbigcdot$};
\node[mark size=2pt, color=gray] at (2.5,2) {$\bbigcdot$};
\node[mark size=2pt, color=gray] at (3.5,2) {$\bbigcdot$};
\node[mark size=2pt, color=gray] at (4.5,2) {$\bbigcdot$};
\node[mark size=2pt, color=gray] at (1,2.5) {$\bbigcdot$};
\node[mark size=2pt, color=gray] at (2,2.5) {$\bbigcdot$};
\node[mark size=2pt, color=gray] at (3,2.5) {$\bbigcdot$};
\node[mark size=2pt, color=gray] at (4,2.5) {$\bbigcdot$};
\node[mark size=2pt, color=gray] at (5,2.5) {$\bbigcdot$};
\node[mark size=2pt, color=gray] at (0.5,3) {$\bbigcdot$};
\node[mark size=2pt, color=gray] at (1.5,3) {$\bbigcdot$};
\node[mark size=2pt, color=gray] at (2.5,3) {$\bbigcdot$};
\node[mark size=2pt, color=gray] at (3.5,3) {$\bbigcdot$};
\node[mark size=2pt, color=gray] at (4.5,3) {$\bbigcdot$};
\node[mark size=2pt, color=gray] at (5.5,3) {$\bbigcdot$};
\node[mark size=2pt, color=gray] at (1,3.5) {$\bbigcdot$};
\node[mark size=2pt, color=gray] at (2,3.5) {$\bbigcdot$};
\node[mark size=2pt, color=gray] at (3,3.5) {$\bbigcdot$};
\node[mark size=2pt, color=gray] at (4,3.5) {$\bbigcdot$};
\node[mark size=2pt, color=gray] at (5,3.5) {$\bbigcdot$};
\node[mark size=2pt, color=gray] at (6,3.5) {$\bbigcdot$};


\end{tikzpicture}

\caption{Example of a mixture of constraints of type 4.  This figure is drawn for $k=6, \beta=1$ and depicts the mixture of constraints for $i=2$ in \eqref{eq:sympl_gf_const1} and \eqref{eq:sympl_gf_const2}. }\label{fig:constr6}
\end{figure}


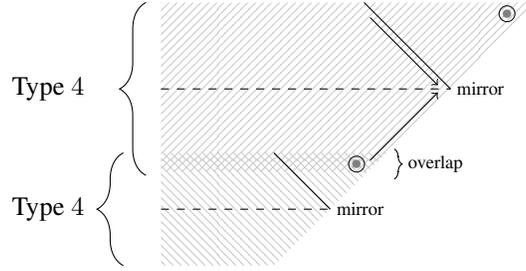
\begin{figure}[!htb]
\centering
\begin{tikzpicture}

\path [pattern=north west lines, pattern color=gray!40] (0,1.5) -- (0,3) -- (3,3) -- (1.5,1.5) -- cycle;
\path [pattern=north east lines, pattern color=gray!40] (0,2.75) -- (0,5) -- (5,5) -- (2.75,2.75)-- cycle;

\draw[] (1.5,3) -- (2.25,2.25);
\draw[dashed] (0,2.25) -- (2.25,2.25);
\draw[] (3.85,3.85) -- (2.7,5);
\draw[dashed] (0,3.85) -- (3.85,3.85);
\node[mark size=1.5pt, color=gray] at (2.6,2.83) {$\bbigcdot$};
\node[mark size=1.5pt, color=gray] at (4.6,4.83) {$\bbigcdot$};
\draw[gray!10!black] (4.6,4.85) circle (3pt);
\draw[gray!10!black] (2.6,2.85) circle (3pt);

\draw[decorate, decoration={brace, amplitude=10pt}]  (-0.2,2.7) -- (-0.2,5) node [midway, xshift=-1.3cm] {\small{Type $4$}};
\draw[decorate, decoration={brace, amplitude=10pt}]  (-0.5,1.5) -- (-0.5,3) node [midway, xshift=-1cm] {\small{Type $4$}};
\draw[decorate, decoration={brace, amplitude=3pt}]  (3.1,3.05) -- (3.1,2.65) node [midway, xshift=0.55cm] {\tiny{overlap}};
\draw[->] (2.78,2.9) -- (3.68,3.8);
\draw[->] (2.78,4.8) -- (3.68,3.9);

\node at (2.65,2.25) {\tiny{mirror}};
\node at (4.25,3.85) {\tiny{mirror}};

\end{tikzpicture}
\caption{Example of combining two split type $4$ constraints.  The dashed horizontal lines represent the lines of reflection, and the solid diagonal lines show where the splitting of the respective types occurs.  The circled elements are the largest element for each section, and the arrows show the location of elements that, due to the interlacing, are forced the be `large', and also direction of growth. }\label{fig:constr10}
\end{figure}







\subsection{Asymptotics at the symmetry point}\label{sec:ks_asympts}

In this subsection we show how the method illustrated above can also be used to recover results of Keating and Snaith on the asymptotics of moments of the characteristic polynomial at the symmetry point, see~\cite{keasna00b}. The original proof involved the Selberg integral and asymptotics for the Barnes G-function. More precisely we show that, for $s\in \mathbb{N}$
\begin{align}
\mathsf{M}_{Sp}(s)\coloneqq\mathbb{E}_{g\in Sp(2N)}\left[\det \left(I-g\right)^s\right]=\mathsf{c}_{Sp}(s)N^{\frac{s(s+1)}{2}}+O_s\left(N^{\frac{s(s+1)}{2}-1}\right),
\end{align}
where the leading order coefficient is explicit:
\begin{align*}
\mathsf{c}_{Sp}(s)=\frac{1}{\prod_{j=1}^{s}(2j-1)!!}.
\end{align*}

By applying Proposition~\ref{BumpGamburdSymplectic} with $x_i\equiv 1$ and inserting this into the combinatorial representation of Definition~\ref{CombinatorialFormulaSymplectic} we obtain the following proposition.

\begin{proposition}
Let $s\in \mathbb{N}$. $\mathsf{M}_{Sp}(s)$ is equal to the cardinality of the set $SP_{\langle N^{s}\rangle}$, namely the number of $(2s)$-symplectic Gelfand-Tsetlin patterns with top row $\langle N^{s}\rangle$.
\end{proposition}

As before, the form of the top row fixes the top right triangle of the pattern, see Figure~\ref{fig:ks_asympt1}. An analogous argument to that given in Proposition~\ref{CombRepSymp2} yields the following.

\begin{proposition}
Let $s\in \mathbb{N}$. Then,
\begin{align*}
\mathsf{M}_{Sp}(s)=N^{\frac{s(s+1)}{2}}\mathsf{vol}\left(\mathsf{V}_{Sp}(s)\right)+O_s\left(N^{\frac{s(s+1)}{2}-1}\right)
\end{align*}
where the set $\mathsf{V}_{Sp}(s)\subset [0,1]^{\frac{s(s+1)}{2}}$ consists of joining two continuous half patterns of length $s$ at the top row, as in the Figure~\ref{fig:ks_asympt2}.
\end{proposition}

Thus, it suffices to show that the volume of $\mathsf{V}_{Sp}(s)$ can be computed explicitly and equals $\mathsf{c}_{Sp}(s)$. We require the following lemma (which is certainly well-known but we have not located this exact form in the literature).

\begin{lemma}\label{ks_lemma1}
Let $s\in \mathbb{N}$. The volume of a continuous half pattern of length $s$ with non-negative coordinates and top row $\left(x_1,\dots,x_{\lfloor \frac{s+1}{2}\rfloor}\right)\in \mathsf{W}_{\lfloor \frac{s+1}{2}\rfloor}^+$, that we denote by $\mathsf{vol}_s\left(x_1,x_2,\dots,x_{\lfloor \frac{s+1}{2}\rfloor}\right)$, is given by:
\begin{align*}
\mathsf{vol}_s\left(x_1,x_2,\dots,x_{\lfloor \frac{s+1}{2}\rfloor}\right)=\prod_{j=1}^{s}\frac{1}{(j-1)!!}\det\left(x_{\lfloor \frac{s+1}{2}\rfloor+1-i}^{2(j-1)+\mathbf{1}(s \ \textnormal{even})}\right)_{i,j=1}^{\lfloor \frac{s+1}{2}\rfloor}.
\end{align*}
\end{lemma}

\begin{proof}
Direct computation by induction on $s$, using multi-linearity of the determinant.
\end{proof}

We finally have:

\begin{proposition}
Let $s\in \mathbb{N}$. Then,
\begin{align*}
\mathsf{vol}\left(\mathsf{V}_{Sp}(s)\right)=\frac{1}{\prod_{j=1}^{s}\left(2j-1\right)!!}.
\end{align*}
\end{proposition}

\begin{proof}
Recall that, see Figure~\ref{fig:ks_asympt2}, $\mathsf{V}_{Sp}(s)$ is obtained by joining at the top row two continuous half patterns with coordinates in $[0,1]$. We then calculate using Lemma~\ref{ks_lemma1} and Andreief's identity:
\begin{align*}
\mathsf{vol}\left(\mathsf{V}_{Sp}(s)\right)&=\int_{1\ge x_1\ge x_2\ge \dots \ge x_{\lfloor \frac{s+1}{2}\rfloor}\ge 0}^{}\mathsf{vol}_s\left(x_1,x_2,\dots,x_{\lfloor \frac{s+1}{2}\rfloor}\right)^2dx_1\cdots dx_{\lfloor \frac{s+1}{2}\rfloor}\\
&=\prod_{j=1}^{s}\left(\frac{1}{(j-1)!!}\right)^2\det\left(\int_{0}^{1}x^{2(i-1)+2(j-1)+2\mathbf{1}(s\ \textnormal{even})}dx\right)^{\lfloor \frac{s+1}{2}\rfloor}_{i,j=1}\\
&=\prod_{j=1}^{s}\left(\frac{1}{(j-1)!!}\right)^2\det\left(\frac{1}{2\left(i+j-\frac{3}{2}+\mathbf{1}(s\ \textnormal{even})\right)}\right)^{\lfloor \frac{s+1}{2}\rfloor}_{i,j=1}.
\end{align*}
In order to evaluate this further one uses the Cauchy determinant formula:
\begin{align*}
\det\left(\frac{1}{x_i-y_j}\right)_{i,j=1}^n=\frac{\prod_{i=2}^{n}\prod_{j=1}^{i-1}(x_i-x_j)(y_j-y_i)}{\prod_{i=1}^{n}\prod_{j=1}^{n}(x_i-y_j)}.
\end{align*}
Applying this with,
\begin{align*}
x_i=2i-\frac{3}{2}+\mathbf{1}(s\ \textnormal{even}), \ y_j=-2j+\frac{3}{2}-\mathbf{1}(s\ \textnormal{even})
\end{align*}
and after some elementary manipulations we readily obtain the statement of the proposition.
\end{proof}

\begin{remark}
Similar arguments apply in the setting of $SO(2N)$, see~\cite{keasna00b} for the original proof.
\end{remark}


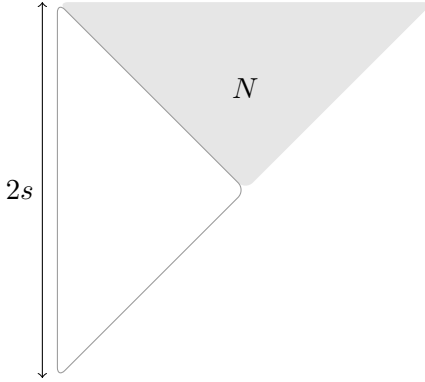
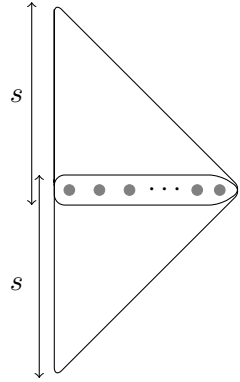
\begin{figure}[!htb]
\centering
\begin{subfigure}[t]{0.48\textwidth}
\centering
\begin{tikzpicture}

\path [rounded corners, fill=gray!20] (0,2.5) -- (2.5,0) -- (5,2.5) -- cycle;
\path [rounded corners, draw=gray!80] (0,2.5) -- (0,-2.5) -- (2.5,0) -- cycle;
\draw[<->] (-0.2,2.5) -- (-0.2,-2.5);
\node at (2.5,1.35) {\textbf{$N$}};
\node at (-0.5,0) {\textbf{$2s$}};

\end{tikzpicture}

\caption{Figure showing a $(2s)$-symplectic Gelfand-Tsetlin pattern with top row (and hence top right triangle) fixed to be $\langle N^{2s}\rangle$.}\label{fig:ks_asympt1}
\end{subfigure}
\hfill
\begin{subfigure}[t]{0.48\textwidth}
\centering
\begin{tikzpicture}

\path [draw=black,rounded corners] (0,2.5) -- (2.5,0) -- (2.2,-0.2) -- (0,-0.2 -- (0,0) -- cycle;
\path [rounded corners, draw=black] (0,2.5) -- (0,-2.5) -- (2.5,0) -- (2.2,0.2) -- (0,0.2) -- (0,0) -- cycle;

\node at (1.5,0) {$\cdots$};

\draw[fill=gray, color=gray] (0.2,0) circle (2pt);
\draw[fill=gray, color=gray] (0.6,0) circle (2pt);
\draw[fill=gray, color=gray] (1,0) circle (2pt);
\draw[fill=gray, color=gray] (1.9,0) circle (2pt);
\draw[fill=gray, color=gray] (2.2,0) circle (2pt);
\draw[<->] (-0.3,2.5) -- (-0.3,-0.2);
\draw[<->] (-0.2,0.2) -- (-0.2,-2.5);
\node at (-0.5,-1.25) {\textbf{$s$}};
\node at (-0.5,1.25) {\textbf{$s$}};
\end{tikzpicture}
\caption{Figure showing the two continuous half patterns in $[0,1]$ joined at the top row which give $\mathsf{V}_{Sp}(s)$.}\label{fig:ks_asympt2}
\end{subfigure}
\caption{Figures showing both the general structure of the (discrete) symplectic half pattern, and the two continuous half patterns formed by the free coordinates joined at the top row.}
\end{figure}

%
%
%


\section{Results for the special orthogonal group $SO(2N)$}\label{sec:orthog}  

We now give the proof of the asymptotic growth of the moments of the moments for $SO(2N)$.  The key difference between the argument presented here and that of Section~\ref{sec:sympl} is that the leading elements in the odd rows of the half-patterns, the `odd-starters', are now allowed to be positive or negative.  This introduces an additional level of complexity due to the fact that now the constraints are not linear (they involve absolute values and signs).

Analogously to the symplectic case outlined in Section~\ref{sec:sympl}, we break the proof down in to steps.  Firstly we prove a proposition connecting the moments of moments to a count of restricted orthogonal Gelfand-Tsetlin patterns. Secondly, we note that the constraints on the patterns fix a triangular region, thus the count simplifies down to considering a subregion of the array. This induces a natural bijection between these constrained patterns and certain integer arrays.  Finally, by considering the number of fixed parameters and moving to a continuous setting, we may apply Theorem~\ref{LatticePointCountTheorem} to achieve Theorem~\ref{MainTheoremOrthogonal}. 

\subsection{A combinatorial representation}

The relevant combinatorial representation for the orthogonal group $SO(2N)$ is the following.

\begin{proposition}\label{CombRepOrthog1}
 Let $k,\beta\in\mathbb{N}$.  Then $\mom_{SO(2N)}(k,\beta)$ is equal to the number of $(4k\beta-1)$-orthogonal Gelfand-Tsetlin patterns $P=\left(\lambda^{(i)}\right)_{i=1}^{4k\beta-1}$ with top row either $\lambda^{(4k\beta-1)}=\langle N^{2k\beta}\rangle$ or $\lambda^{(4k\beta-1)}=\langle N^{2k\beta}\rangle^-$, which moreover satisfy each of the following $k$ constraints for $i=1,\dots,k$:

\begin{align}\label{orthog_constr}
  \sum_{j=(2i-2)\beta+1}^{(2i-1)\beta}&\Sgn(\lambda_j^{(2j-1)})\Sgn(\lambda_{j-1}^{(2j-3)})\left[\sum_{l=1}^{j}|\lambda_l^{(2j-1)}|-2\sum_{l=1}^{j-1}|\lambda_l^{(2j-2)}|+\sum_{l=1}^{j-1}|\lambda_l^{(2j-3)}|\right]\\
  &=\sum_{j=(2i-1)\beta+1}^{2i\beta}\Sgn(\lambda_j^{(2j-1)})\Sgn(\lambda_{j-1}^{(2j-3)})\left[\sum_{l=1}^{j}|\lambda_l^{(2j-1)}|-2\sum_{l=1}^{j-1}|\lambda_l^{(2j-2)}|+\sum_{l=1}^{j-1}|\lambda_l^{(2j-3)}|\right],\nonumber
\end{align}
where $\lambda^{(0)}, \lambda^{(-1)}\equiv 0$. $GT_{SO}(N;k;\beta)$ denotes the set of such patterns.  Further, we write $GT^+_{SO}(N;k;\beta)$ for the set of such constrained $(4k\beta-1)$-orthogonal patterns with top row $\langle N^{2k\beta}\rangle$, and $GT^-_{SO}(N;k;\beta)$ for the equivalent (but disjoint) set with top row $\langle N^{2k\beta}\rangle^{-}$. 
\end{proposition}

\begin{proof}
The proof of Proposition~\ref{CombRepOrthog1} follows entirely the same method as described in the proof of Proposition~\ref{CombRepSymp1}. 
\end{proof}

The case for $k=\beta=1$ is separate from the general case. This is essentially due to the fact that in this particular situation, the limited number of non-fixed elements in the pattern means that the constraints~\eqref{orthog_constr} behave differently compared to the case for higher $k, \beta$ (note that in the case of $GT^+_{SO}(N;1;1)$ the corresponding constraint does not fix any coordinate, as we see in the proof below).  We handle this special case here.

\begin{proposition}\label{orthog_k1b1_prop}
We have that \[\mom_{SO(2N)}(1,1)=2(N+1).\]

\end{proposition}
\begin{proof} 
By Proposition~\ref{CombRepOrthog1}, 
\[\mom_{SO(2N)}(1,1)=|GT^+_{SO}(N;1;1)| + |GT^-_{SO}(N;1;1)|,\]
where here $GT_{SO}(N;1;1)$ is the set of all $(3)$-orthogonal Gelfand-Tsetlin patterns $P$ with top row either $(N,N)$ or $(N,-N)$, corresponding to the sets $GT^+_{SO}(N;1;1)$ and $GT^-_{SO}(N;1;1)$ respectively, satisfying the constraint: 
\begin{equation}\label{orthog_k1b1}
\Sgn(\lambda_1^{(1)})\lambda_1^{(1)}=\Sgn(\lambda_2^{(3)})\Sgn(\lambda_1^{(1)})\lambda_1^{(1)},
\end{equation}
see Figure~\ref{fig:simplest_ortho}. The fact that there is only one `free' parameter, namely $\lambda_1$, here is the key difference between this special case, and the situation for general $k, \beta$. Hence, $|GT^+_{SO}(N;1;1)| = 2N+1$ since all values of $0\leq |\lambda_1^{(1)}| \leq N$ are valid.   However, the only option satisfying constraint~\eqref{orthog_k1b1} in the second case is $\lambda_1^{(1)}\equiv0$.  Thus,
\[\mom_{SO(2N)}(1,1)=2(N+1).\]


\begin{figure}[!htb]
\centering
\begin{tikzpicture}[scale=0.8, every node/.style={scale=0.8}]
\node at (1.2,2.7) {\Large{$P\in GT^{+}_{SO}(N;1;1)$ }};
\draw[thin] (-0.9,2.4) -- (3.1,2.4);
\draw[thin] (-0.9,2.35) -- (3.1,2.35);
\node at (0,0) {\Large{$\lambda_1^{(1)}$}};
\node at (1,1) {\Large{$N$}};
\node at (0,2) {\Large{$N$}};
\node at (2,2) {\Large{$N$}};

\node at (6.5,2.7) {\Large{$Q\in GT^{-}_{SO}(N;1;1)$ }};
\draw[thin] (4.5,2.4) -- (8.5,2.4);
\draw[thin] (4.5,2.35) -- (8.5,2.35);
\node at (5.5,0) {\Large{$\lambda_1^{(1)}$}};
\node at (6.5,1) {\Large{$N$}};
\node at (5.5,2) {\Large{$-N$}};
\node at (7.5,2) {\Large{$N$}};

\end{tikzpicture}

\caption{Cases for determining $\mom_{SO(2N)}(1,1)$. The relevant constraint is $\lambda_1^{(1)}=\lambda_1^{(1)}\cdot \Sgn(\pm N)$.}\label{fig:simplest_ortho}
\end{figure}
\end{proof}

Henceforth we assume that we are in the general case (i.e. we exclude the case $k=\beta=1$).  Then, we note that by requiring the top row of the pattern $P$ to be either $\langle N^{2k\beta}\rangle$ or $\langle N^{2k\beta}\rangle^-$, the top right triangle of $GT_{SO}(N;k,\beta)$ is also determined, as shown in Figure~\ref{fig:orthofixed}.  We now introduce notation which captures the sign of the odd-starters for a given pattern $P\in GT_{SO}(N;k;\beta)$.  Note that the ability of the odd-starters to be positive or negative is one of the key differences between the orthogonal and the symplectic case.


We consider the following decomposition of $GT_{SO}(N;k;\beta)$ into the disjoint union:
\[GT_{SO}(N;k;\beta) = \bigcup_{\underline{\varepsilon}\in\{\pm1\}^{2k\beta}} GT^{\underline{\varepsilon}}_{SO}(N;k;\beta),\]
where $GT^{\underline{\varepsilon}}_{SO}(N;k;\beta)$ is the subset of $GT_{SO}(N;k;\beta)$ where the sign of $\lambda_{i}^{(2i-1)}$ for $1\leq i \leq 2k\beta$ is required to be equal to $\varepsilon_i$.  We decompose in this way due to the requirement of convexity in Theorem~\ref{LatticePointCountTheorem}. One then sees that, for instance,
\[GT^+_{SO}(N;k;\beta) = \bigcup_{\substack{\underline{\varepsilon}\in\{\pm1\}^{2k\beta}:\\ \varepsilon_{2k\beta}=1}} GT^{\underline{\varepsilon}}_{SO}(N;k;\beta).\]
Further examples of the definition are given by Figure~\ref{fig:oddstarter_eg}.

As in Section~\ref{sec:sympl}, for ease we now concentrate on the undetermined elements.  The following definition formally defines a relabelling of said parts, and Figure~\ref{fig:orthobij} demonstrates the bijection between a given pattern $P\in GT^{\underline{\varepsilon}}_{SO}(N;k;\beta)$ and the renaming.  In spirit, this process is the same as that described in Definition~\ref{def:sympl_int_array}, though with the added complexity of the signs of the odd-starters. 

\begin{definition}\label{def:ortho_int_array}
We consider the decomposition of $\mathfrak{I}_{SO}(N;k;\beta)$ into the union of multisets
\[\mathfrak{I}_{SO}(N;k;\beta)=\bigcup_{\underline{\varepsilon}\in\{\pm 1\}^{2k\beta}}\mathfrak{I}^{\underline{\varepsilon}}_{SO}(N;k;\beta),\]
where for a fixed $\underline{\varepsilon}\in\{\pm1\}^{2k\beta}$, $\mathfrak{I}^{\underline{\varepsilon}}_{SO}(N;k;\beta)$ is the set of integer arrays $(y^{(i)})_{i=1}^{4k\beta-3}\in\mathbb{Z}^{k\beta(2k\beta-1)}$ satisfying the following additional requirements:
\begin{enumerate}
\item $y^{(i)},y^{(4k\beta-2-i)}\in \mathsf{S}^+_{\lfloor \frac{i+1}{2}\rfloor}$ for $1\le i \le 2k\beta-1$,
\item both $(y^{(i)})_{i=1}^{2k\beta-1}$ and $(y^{(4k\beta-2-i)})_{i=1}^{2k\beta-1}$ form $(2k\beta-1)$-orthogonal Gelfand-Tsetlin patterns,
\item $0\le y_j^{(i)} \le N$ for any valid $i,j$, unless $y_j^{(i)}$ is an odd-starter with its associated $\varepsilon_{(i+1)/2}=-1$, in which case $0<y_j^{(i)}\leq N$,
\item the rows  $(y^{(i)})_{i=1}^{4k\beta-3}$ fulfil the following constraints:

In the case $k$ is even, let $i=1,\dots,\frac{k}{2}$ (with $y^{(-1)},y^{(0)},y^{(4k\beta-2)},y^{(4k\beta-1)}\equiv 0$, and $\varepsilon_0\equiv 1$).  Then,
\begin{align}\label{eq:ortho_gf_const1}
\sum_{j=(2i-2)\beta+1}^{(2i-1)\beta}&\varepsilon_j\varepsilon_{j-1}\left[\sum_{l=1}^{j}y_l^{(2j-1)}-2\sum_{l=1}^{j-1}y_l^{(2j-2)}+\sum_{l=1}^{j-1}y_l^{(2j-3)}\right]\\
&=\sum_{j=(2i-1)\beta+1}^{2i\beta}\varepsilon_j\varepsilon_{j-1}\left[\sum_{l=1}^{j}y_l^{(2j-1)}-2\sum_{l=1}^{j-1}y_l^{(2j-2)}+\sum_{l=1}^{j-1}y_l^{(2j-3)}\right]\nonumber
\end{align}
and
\begin{align}\label{eq:ortho_gf_const2}
\sum_{j=(2i-2)\beta+1}^{(2i-1)\beta}&\varepsilon_{2k\beta-j+1}\varepsilon_{2k\beta-j}\left[\sum_{l=1}^{j}y_l^{(4k\beta-2j-1)}-2\sum_{l=1}^{j-1}y_l^{(4k\beta-2j)}+\sum_{l=1}^{j-1}y_l^{(4k\beta-2j+1)}\right]\\
&=\sum_{j=(2i-1)\beta+1}^{2i\beta}\varepsilon_{2k\beta-j+1}\varepsilon_{2k\beta-j}\left[\sum_{l=1}^{j}y_l^{(4k\beta-2j-1)}-2\sum_{l=1}^{j-1}y_l^{(4k\beta-2j)}+\sum_{l=1}^{j-1}y_l^{(4k\beta-2j+1)}\right].\nonumber
\end{align}
While, when $k$ is odd we have the same constraints as above for $i=1,\dots, \frac{k-1}{2}$ along with:
\begin{align}\label{eq:ortho_gf_const3}
\sum_{j=(k-1)\beta+1}^{k\beta}&\varepsilon_{j}\varepsilon_{j-1}\left[\sum_{l=1}^{j}y_l^{(2j-1)}-2\sum_{l=1}^{j-1}y_l^{(2j-2)}+\sum_{l=1}^{j-1}y_l^{(2j-3)}\right]\\
&=\sum_{j=(k-1)\beta+1}^{k\beta}\varepsilon_{2k\beta-j+1}\varepsilon_{2k\beta-j}\left[\sum_{l=1}^{j}y_l^{(4k\beta-2j-1)}-2\sum_{l=1}^{j-1}y_l^{(4k\beta-2j)}+\sum_{l=1}^{j-1}y_l^{(4k\beta-2j+1)}\right].\nonumber
\end{align} 
\end{enumerate}
Observe that, as in the symplectic case, for both $k$ odd and even there are a total of $k$ constraints.
\end{definition}

Then, analogously to how $\mathfrak{B}_{Sp}$  was defined in Section~\ref{sec:sympl}, (see \eqref{eq:symbij}), one may also define 
\begin{equation}\label{eq:orthobij}
\mathfrak{B}_{SO}:GT_{SO}(N;k;\beta) \longrightarrow \mathfrak{I}_{SO}(N;k;\beta).
\end{equation}

The bijection is depicted by Figure~\ref{fig:orthobij}, and can be constructed as follows. Take $P\in GT_{SO}(N;k;\beta)$ so $P=(\lambda^{(i)})_{i=1}^{4k\beta-1}$.  In particular, there exists $\underline{\varepsilon}\in\{\pm1\}^{2k\beta}$ such that $P\in GT^{\underline{\varepsilon}}_{SO}(N;k;\beta)$.  Due to the interlacing $\lambda^{(4k\beta-3)}\prec \langle N^{2k\beta-1}\rangle=\lambda^{(4k\beta-2)}$, all but one element of $\lambda^{(4k\beta-3)}$ is fixed:
\begin{align*}
\lambda_1^{(4k\beta-3)},\dots,\lambda_{2k\beta-2}^{(4k\beta-3)}\equiv N,\\
0 \le |\lambda_{2k\beta-1}^{(4k\beta-3)}| \le N.
\end{align*}

We now set $y_1^{(4k\beta-3)}=|\lambda_{2k\beta-1}^{(4k\beta-3)}|$ and $\varepsilon_{2k\beta-1}=\Sgn(\lambda_{2k\beta-1}^{(4k\beta-3)})$. Repeating the same logic, we consider the next pair of interlaced rows $\lambda^{(4k\beta-4)}\prec \lambda^{(4k\beta-3)}$ which once more fixes all but one coordinate:
\begin{align*}
\lambda_1^{(4k\beta-4)},\dots,\lambda_{2k\beta-3}^{(4k\beta-4)}\equiv N,\\
y_1^{(4k\beta-3)} = |\lambda_{2k\beta-1}^{(4k\beta-3)}| \le \lambda_{2k\beta-2}^{(4k\beta-4)} \le N.
\end{align*}

Thus set $y_1^{(4k\beta-4)}=\lambda_{2k\beta-2}^{(4k\beta-4)}$.  This process can be repeated up to and including $\lambda^{(2k\beta)}$, after which there are no more coordinates fixed by the interlacing.   Thereafter set $y_j^{(i)}=|\lambda_{j}^{(i)}|$, and throughout use the fact that $\varepsilon_j=\Sgn(\lambda_{j}^{(2j-1)})$.  It is apparent that this entire process is invertible, hence the map given by this construction, $\mathfrak{B}_{SO}$ is a bijection.  We may then employ Proposition~\ref{CombRepOrthog1} to achieve the following statement. 

\begin{proposition}\label{CombRepOrthog2}
Let $k, \beta\in\mathbb{N}$. Then 
\begin{align*}
\mom_{SO(2N)}(k,\beta)=\#GT_{SO}(N;k,\beta)&=\sum_{\underline{\varepsilon}\in\{\pm 1\}^{2k\beta}}\#GT^{\underline{\varepsilon}}_{SO}(N;k;\beta)\\
&= \sum_{\underline{\varepsilon}\in\{\pm 1\}^{2k\beta}}\#\mathfrak{I}^{\underline{\varepsilon}}_{SO}(N;k;\beta)\\
&=\# \mathfrak{I}_{SO}(N;k,\beta).
\end{align*}
\end{proposition} 


\begin{figure}[!htb]
\centering
\begin{tikzpicture}[scale=0.6, every node/.style={scale=0.6}]

\draw[decorate, decoration={brace, amplitude=10pt}] (-0.6,10.4) -- (10.7,10.4) node [midway, yshift=.9cm] {$2k\beta$};

\path [draw=gray!10, fill=gray!10] (-0.6,10.2) -- (10.7,10.2) -- (5,4.3) -- cycle;
\node at (0,10) {$N$};
\node at (2,10) {$\cdots$};
\node at (6,10) {$\cdots$};
\node at (8,10) {$N$};
\node at (10,10) {$N$};
\node at (1,9) {$N$};
\node at (3,9) {$\cdots$};
\node at (5,9) {$\cdots$};
\node at (7,9) {$N$};
\node at (9,9) {$N$};
\node at (8,8) {$N$};
\node at (2,8) {$N$};
\node at (4,8) {$\cdots$};
\node at (6,8) {$N$};
\node at (8,8) {$N$};
\node at (7,7) {$N$};
\draw (6,6) node[rotate=80] {$\ddots$};
\draw (5,7) node[rotate=-10] {$\ddots$};
\draw (3,7) node[rotate=-10] {$\ddots$};
\draw (4,6) node[rotate=-10] {$\ddots$};
\draw (3,5) node[rotate=-10] {$\ddots$};
\node at (5,5) {$N$};

\draw[thin, color=gray!60] (-0.6,10.2) -- (-0.6,-1.1);
\draw[thin, color=gray!60] (5,4.3) -- (-0.6,-1.1);
\draw[thin, color=gray!60] (5,4.3) -- (-0.6,10.2);
\draw[decorate, decoration={brace, amplitude=10pt}]  (-0.7,-0.3) -- (-0.7,8.3) node [midway, xshift=-1.4cm] {$2k\beta-1$};
\node at (-2.1, 3.6) {odd-starters};
\node at (0,8) {$\ast$};
\node at (1,7) {$\ast$};
\node at (0,6) {$\ast$};
\node at (2,6) {$\ast$};
\node at (4,4) {$\ast$};
\draw (1,5) node[rotate=-10] {$\ddots$};
\node at (0,4) {$\vdots$};
\node at (0,2) {$\vdots$};
\draw (1,3) node[rotate=-10] {$\ddots$};
\draw (2,2) node[rotate=80] {$\ddots$};
\node at (2,4) {$\ast$};
\node at (3,3) {$\ast$};
\node at (1,1) {$\ast$};
\node at (0,0) {$\ast$};
\end{tikzpicture}

\caption{Figure depicting the fixed region of a $(4k\beta-1)$-orthogonal Gelfand-Tsetlin pattern with top row $\langle N^{2k\beta}\rangle$.  The shaded area represents the fixed region, whilst the unshaded region shows which elements have some freedom in the values that they can take.}\label{fig:orthofixed}
\end{figure}
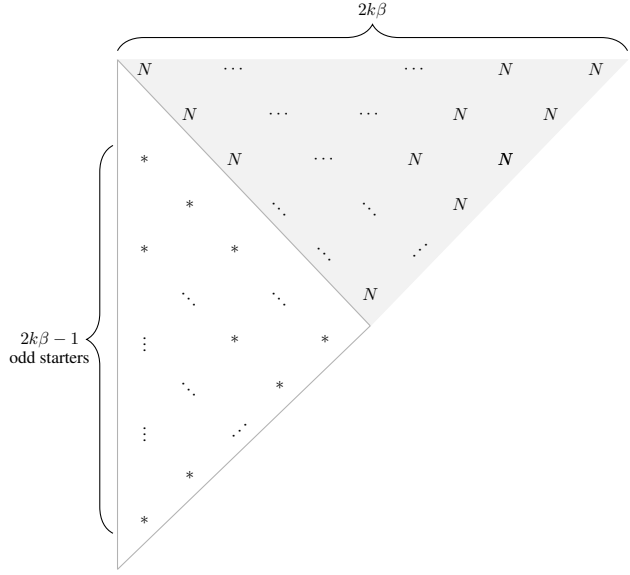


\begin{figure}[!htb]
\centering
\begin{tikzpicture}[scale=0.6, every node/.style={scale=0.6}]
\node at (3.5,6.7) {\Large{$P\in GT^{\underline{\varepsilon}}_{SO}(8;2;1)$ }};
\node at (3.5,7.5) {\Large{$\underline{\varepsilon}=(-1,1,-1,-1)$}};
\draw[thin] (-.5,6.4) -- (6.5,6.4);
\draw[thin] (-.5,6.35) -- (6.5,6.35);
\node at (0,0) {\Large{$-1$}};
\node at (1,1) {\Large{$2$}};
\node at (0,2) {\Large{$0$}};
\node at (2,2) {\Large{$6$}};
\node at (1,3) {\Large{$5$}};
\node at (3,3) {\Large{$8$}};
\node at (0,4) {\Large{$-2$}};
\node at (2,4) {\Large{$8$}};
\node at (4,4) {\Large{$8$}};
\node at (1,5) {\Large{$8$}};
\node at (3,5) {\Large{$8$}};
\node at (5,5) {\Large{$8$}};
\node at (0,6) {\Large{$-8$}};
\node at (2,6) {\Large{$8$}};
\node at (4,6) {\Large{$8$}};
\node at (6,6) {\Large{$8$}};

\node at (12.5,6.7) {\Large{$Q\in GT^{\underline{\varepsilon}}_{SO}(5;2;1)$ }};
\node at (12.5,7.5) {\Large{$\underline{\varepsilon}=(-1,-1,1,1)$}};
\draw[thin] (8.5,6.4) -- (15.5,6.4);
\draw[thin] (8.5,6.35) -- (15.5,6.35);
\node at (9,0) {\Large{$-3$}};
\node at (10,1) {\Large{$2$}};
\node at (9,2) {\Large{$-3$}};
\node at (11,2) {\Large{$4$}};
\node at (10,3) {\Large{$3$}};
\node at (12,3) {\Large{$5$}};
\node at (9,4) {\Large{$2$}};
\node at (11,4) {\Large{$5$}};
\node at (13,4) {\Large{$5$}};
\node at (10,5) {\Large{$5$}};
\node at (12,5) {\Large{$5$}};
\node at (14,5) {\Large{$5$}};
\node at (9,6) {\Large{$5$}};
\node at (11,6) {\Large{$5$}};
\node at (13,6) {\Large{$5$}};
\node at (15,6) {\Large{$5$}};
\end{tikzpicture}

\caption{Examples of patterns $P, Q$ in $GT^{\underline{\varepsilon}}_{SO}(N;k;\beta)$ for $k=2$, $\beta=1$, and different, given values of $N$ and $\underline{\varepsilon}$.}\label{fig:oddstarter_eg}
\end{figure}

\begin{figure}[!htb]
\centering
\begin{tikzpicture}[scale=0.6, every node/.style={scale=0.6}]

\path [draw=gray!10, fill=gray!10] (-0.8,10.3) -- (10.8,10.3) -- (5,4.5) -- cycle;
\node at (0,10) {$\pm N$};
\node at (4,10) {$\cdots$};
\node at (2,10) {$\cdots$};
\node at (6,10) {$N$};
\node at (8,10) {$N$};
\node at (10,10) {$N$};
\node at (1,9) {$N$};
\node at (3,9) {$\cdots$};
\node at (5,9) {$\cdots$};
\node at (7,9) {$N$};
\node at (9,9) {$N$};
\node at (8,8) {$N$};
\node at (2,8) {$N$};
\node at (4,8) {$\cdots$};
\node at (6,8) {$N$};
\node at (8,8) {$N$};
\node at (7,7) {$N$};
\draw (6,6) node[rotate=80] {$\ddots$};
\draw (5,7) node[rotate=-10] {$\ddots$};
\draw (3,7) node[rotate=-10] {$\ddots$};
\draw (4,6) node[rotate=-10] {$\ddots$};
\draw (3,5) node[rotate=-10] {$\ddots$};
\node at (5,5) {$N$};


\node at (0,8) {$\lambda_{2k\beta-1}^{(4k\beta-3)}$};
\node at (1,7) {$\lambda_{2k\beta-2}^{(4k\beta-4)}$};
\node at (0,6) {$\lambda_{2k\beta-2}^{(4k\beta-5)}$};
\node at (2,6) {$\lambda_{2k\beta-3}^{(4k\beta-5)}$};
\draw (3,5) node[rotate=-10] {$\ddots$};
\node at (4,4) {$\lambda_{1}^{(2k\beta-1)}$};
\draw (1,5) node[rotate=-10] {$\ddots$};
\node at (0,4.5) {$\vdots$};
\node at (0,2.5) {$\vdots$};
\draw (1,3) node[rotate=-10] {$\ddots$};
\draw (2,2) node[rotate=80] {$\ddots$};
\node at (2,4) {$\lambda_{2}^{(2k\beta-1)}$};
\node at (3,3) {$\lambda_{1}^{(2k\beta-2)}$};
\node at (1,1) {$\lambda_{1}^{(2)}$};
\node at (0,0) {$\lambda_{1}^{(1)}$};

\node at (13,9) {};
\node at (15,9) {};
\node at (17,9) {};
\node at (14,8) {};
\node at (16,8) {};
\node at (18,8) {};
\node at (18,6) {};
\draw (17,7) node[rotate=-10] {};
\draw (15,7) node[rotate=-10] {};
\draw (16,6) node[rotate=-10] {};
\draw (15,5) node[rotate=-10] {};
\node at (17,5) {};

\node at (12,10) {$\varepsilon_{2k\beta}$};
\node at (12,8) {$\varepsilon_{2k\beta-1}y_1^{(4k\beta-3)}$};
\node at (13,7) {$y_1^{(4k\beta-4)}$};
\node at (12,6) {$\varepsilon_{2k\beta-2}y_2^{(4k\beta-5)}$};
\node at (14,6) {$y_1^{(4k\beta-5)}$};
\node at (16,4) {$y_{1}^{(2k\beta-1)}$};
\draw (13,5) node[rotate=-10] {$\ddots$};
\draw (15,5) node[rotate=-10] {$\ddots$};
\node at (12,4.5) {$\vdots$};
\node at (12,2.5) {$\vdots$};
\draw (13,3) node[rotate=-10] {$\ddots$};
\draw (14,2) node[rotate=80] {$\ddots$};
\node at (14,4) {$y_{2}^{(2k\beta-1)}$};
\node at (15,3) {$y_{1}^{(2k\beta-2)}$};
\node at (13,1) {$y_{1}^{(2)}$};
\node at (12,0) {$\varepsilon_{1}y_{1}^{(1)}$};

\node at (9,5) {\huge{$\xrightarrow[]{\mathfrak{B}_{SO}}$}};

\end{tikzpicture}

\caption{Pictorial representation of the relabelling of the coordinates given by the bijection $\mathfrak{B}^{\underline{\varepsilon}}_{SO}:GT^{\underline{\varepsilon}}_{SO}(N;k;\beta) \longrightarrow \mathfrak{I}^{\underline{\varepsilon}}_{SO}(N;k;\beta)$.  Above on the right hand side (the image of the bijection), $\varepsilon_j=\Sgn(\lambda^{(2j-1)}_j)$ for $j=1,\dots,2k\beta-1$ and $\varepsilon_{2k\beta}=\Sgn(\lambda^{(4k\beta-1)}_{2k\beta})=\Sgn(\pm N)$.}\label{fig:orthobij}
\end{figure}

\subsection{Asymptotics and the leading order coefficient}\label{sec:ortho_asympts}

Recall, from Section~\ref{sec:sympl_asympt}, that we defined continuous half-patterns and continuous orthogonal Gelfand-Tsetlin patterns using the continuous Weyl chamber, 
\begin{align*}
\mathsf{W}_N=\{ x=(x_1,\dots,x_N)\in \mathbb{R}^N:x_1\ge \dots \ge x_N\}.
\end{align*}

There we defined the index set $\mathcal{S}^{Sp}_{(k,\beta)}$, here we give the equivalent definition for the orthogonal case.  For more explanation of the construction of this set, see the Section~\ref{sec:sympl_asympt}.

\begin{align*}
\mathcal{S}^{SO}_{(k,\beta)} \coloneqq \bigg\{ (m,n):\ &1 \le m \le \bigg\lfloor \frac{n+1}{2} \bigg\rfloor\text{ and }1 \le n \le 2k\beta-1; \\
&\text{or }1 \le i \le \bigg\lfloor \frac{4k\beta-n-1}{2}\bigg\rfloor\text{ and }2k\beta\le n < 4k\beta-3;\\
 &n\neq 4\beta-1, 8\beta-1, \dots, 4(k-1)\beta-1 \bigg\}\\
 \cup \bigg\{ (m,4n&\beta-1): \ 1\le m \le 2n\beta-1 \text{ and }1\le n\le \left\lfloor \frac{k}{2} \right\rfloor;\\
 &\text{or }1\le m \le 2(k-n)\beta-1 \text{ and }\left\lfloor\frac{k}{2}\right\rfloor+1\le n< k \}\bigg\}.
\end{align*}

Note that the size of the set $\mathcal{S}^{SO}_{(k,\beta)}$ is $k\beta (2k\beta-1)-k$.  The set corresponding to the indices `missing' from $\mathcal{S}^{SO}_{(k,\beta)}$ is the following
\[\mathcal{T}^{SO}_{(k,\beta)}\coloneqq \{(m,n) : y_m^{(n)}\in \mathfrak{I}_{SO}(N;k;\beta)\} \backslash \mathcal{S}^{SO}_{(k,\beta)}.\]
Now define the following set $\mathcal{V}_{(k,\beta;\underline{\varepsilon})}^{SO}\subset\mathbb{R}^{k\beta(2k\beta-1)-k}$, which is the continuous version of $\mathfrak{I}^{\underline{\varepsilon}}_{SO}(N;k;\beta)$, except that a particular choice of $k$ of the coordinates from $\mathfrak{I}^{\underline{\varepsilon}}_{SO}(N;k;\beta)$ are determined by the linear equations, \cref{eq:ortho_gf_const1,eq:ortho_gf_const2,eq:ortho_gf_const3}.  Then, $\mathcal{V}_{(k,\beta;\underline{\varepsilon})}^{SO}$  comprises the following elements.   Firstly, we take coordinates $x_m^{(n)}$ indexed by $(m,n)\in \mathcal{S}^{SO}_{(k,\beta)}$ which moreover satisfy the following: 
\begin{align*}
0\le x_m^{(n)} \le 1,\quad&\text{for }(m,n)\in \mathcal{S}^{SO}_{(k,\beta)},
\end{align*}
unless $(m,n)$ denotes the position of an odd-starter with its corresponding $\varepsilon_{(n+1)/2}=-1$ in which case $0<x_m^{(n)}\leq 1$. Take $\underline{\varepsilon}$ just as in the definition of $\mathcal{I}^{\underline{\varepsilon}}_{SO}(N;k,\beta)$, i.e. a fixed set of signs for the odd-starters.
Additionally, $\mathcal{V}_{(k,\beta;\underline{\varepsilon})}^{SO}$ contains the following $k$ elements, determined by the linear equations~\eqref{eq:ortho_gf_const1}--\eqref{eq:ortho_gf_const3} in the definition of $\mathfrak{I}^{\underline{\varepsilon}}_{SO}(N;k;\beta)$, 
\begin{align*}
x_{\lfloor\frac{n+1}{2}\rfloor}^{(n)}\quad&\text{for }n= 4\beta-1, 8\beta-1, \dots, 4\lfloor\tfrac{k}{2}\rfloor\beta-1,\\
x_{\lfloor\frac{4k\beta-n-1}{2}\rfloor}^{(n)}\quad&\text{for }n=4(\lfloor\tfrac{k}{2}\rfloor+1)\beta-1,\dots,4(k-1)\beta-1, 4k\beta-3.
\end{align*}
Thus,
 \begin{itemize}
\item $0 \le x_{m}^{(n)} \le 1$, for all $x_m^{(n)}\in \mathcal{V}_{(k,\beta;\underline{\varepsilon})}^{SO}$, unless $(m,n)$ denotes the position of an odd-starter with its corresponding $\varepsilon_{(n+1)/2}=-1$, in which case $0<x_m^{(n)}\leq 1$,
\item $x^{(n)}, x^{(4k\beta-n)} \in \mathsf{W}_{\lfloor \frac{n+1}{2}\rfloor}^+$, for all $n=1, \dots, 2k\beta-1$,
\item both $(x^{(n)})_{n=1}^{2k\beta-1}$ and $(x^{(4k\beta-n)})_{n=1}^{2k\beta-1}$ form continuous $(2k\beta-1)$-orthogonal Gelfand-Tsetlin patterns.
\end{itemize}

Observe that, just as in the symplectic case, $\mathcal{V}_{(k,\beta;\underline{\varepsilon})}^{SO}$ is convex as an intersection of hyperplanes. Moreover, $\mathcal{V}_{(k,\beta;\underline{\varepsilon})}^{SO}$ is contained in the cube $[0,1]^{k\beta(2k\beta-1)-k}$ and hence in a closed ball of radius $\sqrt{k\beta(2k\beta-1)-k}$.

\begin{proof}[Proof of Theorem \ref{MainTheoremOrthogonal}]
The fact that the moments of moments are polynomials in $N$ was proven in Proposition \ref{polystrucsympl}, and the case of $k=\beta=1$ was handled above in Proposition~\ref{orthog_k1b1_prop}.

What remains to be shown is the statement concerning the leading order for general $k, \beta$.  Firstly note that for a given $\underline{\varepsilon}\in\{\pm1\}^{2k\beta}$:

\begin{align*}
\#\mathfrak{I}^{\underline{\varepsilon}}_{SO}(N;k;\beta)=\# \left(\mathbb{Z}^{k\beta(2k\beta-1)-k}\cap\left(N\mathcal{V}_{(k,\beta;\underline{\varepsilon})}^{SO}\right)\right),
\end{align*}
where for a set $\mathcal{A}$, we write $N \mathcal{A}=\{Nx: x \in \mathcal{A} \}$ for its dilate by a factor of $N$. Making use of Theorem~\ref{LatticePointCountTheorem} with $\mathcal{S}=N\mathcal{V}_{(k,\beta;\underline{\varepsilon})}^{SO}$ we get:
\begin{align*}
\#\mathfrak{I}^{\underline{\varepsilon}}_{SO}(N;k;\beta)&=\mathsf{vol}\left(N\mathcal{V}_{(k,\beta;\underline{\varepsilon})}^{SO}\right)+O_{k,\beta}\left(N^{k\beta(2k\beta-1)-k-1}\right)\\
&=N^{k\beta(2k\beta-1)-k}\mathsf{vol}\left(\mathcal{V}_{(k,\beta;\underline{\varepsilon})}^{SO}\right)+O_{k,\beta}\left(N^{k\beta(2k\beta-1)-k-1}\right).
\end{align*}
Thus, by Proposition~\ref{CombRepOrthog2} we obtain:
\begin{align*}
\mom_{SO(2N)}(k,\beta)&=\sum_{\underline{\varepsilon}\in\{\pm 1\}^{2k\beta}}\#\mathfrak{I}^{\underline{\varepsilon}}_{SO}(N;k;\beta)\\&=\sum_{\underline{\varepsilon}\in\{\pm 1\}^{2k\beta}}\left[N^{k\beta(2k\beta-1)-k}\mathsf{vol}\left(\mathcal{V}_{(k,\beta;\underline{\varepsilon})}^{SO}\right)+O_{k,\beta}\left(N^{k\beta(2k\beta-1)-k-1}\right)\right]\\
&=\mathfrak{c}_{SO}(k,\beta)N^{k\beta(2k\beta-1)-k}+O_{k,\beta}\left(N^{k\beta(2k\beta-1)-k-1}\right)
\end{align*}
where
\begin{align}
\mathfrak{c}_{SO}(k,\beta)=\sum_{\underline{\varepsilon}\in\{\pm 1\}^{2k\beta}}\mathsf{vol}\left(\mathcal{V}_{(k,\beta;\underline{\varepsilon})}^{SO}\right).
\end{align}
It then once more suffices to prove that $\mathfrak{c}_{SO}(k,\beta)>0$, which is the content of Lemma \ref{PositivityOfVolumeOrthogonal} below.
\end{proof}
%

\begin{lemma}\label{PositivityOfVolumeOrthogonal}
Let $k,\beta \in \mathbb{N}$. Then,
\begin{align}
\mathfrak{c}_{SO}(k,\beta)>0.
\end{align} 
\end{lemma}

\begin{proof}

Recall that \[\mathfrak{c}_{SO}(k,\beta)=\sum_{\underline{\varepsilon}\in\{\pm 1\}^{2k\beta}}\mathsf{vol}\left(\mathcal{V}_{(k,\beta;\underline{\varepsilon})}^{SO}\right).\] 
Thus, the proof of the strict positivity of the leading order coefficient $\mathfrak{c}_{SO}(k,\beta)$ can be deduced from showing that, for at least one choice of $\underline{\varepsilon}\in\{\pm1\}^{2k\beta}$, the volume $\mathsf{vol}\left(\mathcal{V}_{(k,\beta;\underline{\varepsilon})}^{SO}\right)$ is strictly positive.  Henceforth, we choose $\underline{\varepsilon}=(1,1,\dots,1)$.  Then, the argument is near identical to the one given in the symplectic case, see the proof of Lemma~\ref{PositivityOfVolumeSymplectic}, aside from trivial differences in the shapes considered. 

\end{proof}


\section{Examples}\label{sec:examples}  
We now give various explicit examples of the polynomials $\mom_{G(N)}(k,\beta)$ for $G(N)\in\{ Sp(2N), SO(2N)\}$ and small, integer values of $k, \beta$.  These examples were calculated using expressions for averages over $Sp(2N), SO(2N)$ using Toeplitz and Hankel determinants, see for example~\cite{garmig19}.  For small $k, \beta$ this is a computationally feasible task, but the complexity grows swiftly with $k, \beta$. 

\subsection{Symplectic case}
\begin{align*}
\mom_{Sp(2N)}(1,1)&=\frac{1}{2}(N+1)(N+2)\\
\mom_{Sp(2N)}(1,2)&=\frac{1}{181440}(N+1) (N+2) (N+3) (N+4) (2 N+5)\\
&\times \left(23 N^4+230 N^3+905 N^2+1650N+1512\right)\\
\mom_{Sp(2N)}(1,3)&=\frac{1}{405483668029440000}(N+1) (N+2) (N+3) (N+4) (N+5) (N+6)\\
&\times \left(10253349 N^{14}+502414101
   N^{13}+11401640999 N^{12}+158831139621 N^{11}\right.\\
   &\qquad+ 1517607151837 N^{10}+10524657547803
   N^9+54662663279397 N^8\\
   &\qquad+ 216189375784263 N^7+655178814761674 N^6+1517469287314596
   N^5\\
   &\qquad+ 2654161159219304 N^4+3424171976788416 N^3+3125457664755840
   N^2\\
   &\qquad\left.+ 1856618315596800 N+563171761152000\right)\\
\mom_{Sp(2N)}(2,1)&=\frac{1}{10080}(N+1) (N+2) (N+3) (N+4) \left(3 N^4+30 N^3+127 N^2+260 N+420\right)\\
\mom_{Sp(2N)}(3,1)&= \frac{1}{133382785536000}(N+1) (N+2) (N+3) (N+4) (N+5) (N+6) \\
&\times (5810 N^{12}+244020 N^{11}+4746259 N^{10}+56513415
   N^9\\
   &\qquad+459233580 N^8+2688408450 N^7+11665223647 N^6+38004428175 N^5\\
   &\qquad+93222284960 N^4+171600705780 N^3+236485094544 N^2\\
   &\qquad+239758263360 N+185253868800)
\end{align*}

\subsection{Orthogonal case}

\begin{align*}
  \mom_{SO(2N)}(1,1)&= 2(N+1)\\
  \mom_{SO(2N)}(1,2)&= \frac{1}{60} (N+1) (N+2) (2 N+3) \left(13 N^2+39 N+20\right)\\
  \mom_{SO(2N)}(1,3)&= \frac{1}{43589145600}(N+1) (N+2) (N+3) (N+4) \\
  &\times\left(677127 N^{10}+16928175 N^9+188303800 N^8+1226849750 N^7+5186281891 N^6\right.\\
  &\qquad+ 14881334615 N^5+29392642150 N^4+39443286500 N^3\\
  &\qquad\left.+\ 34230199032 N^2+17098220160 N+3632428800\right)
\end{align*}
\begin{align*}
  \mom_{SO(2N)}(2,1)&= \frac{1}{2} (N+1)^2 (N+2)^2\\
  \mom_{SO(2N)}(3,1)&= \frac{1}{1360800}(N+1) (N+2)^2 (N+3)^2 (N+4)(N^2+5 N+9)\hspace{4cm}\\
  &\times(31 N^4+310 N^3+1163 N^2+1940N+2100)
\end{align*}

\section{Acknowledgements}\label{sec:acknowledgements}    

We would like to thank Benjamin Fahs for a useful discussion on the results of his paper \cite{fahs19}. ECB is grateful to the Heilbronn Institute for Mathematical Research for support.  TA and JPK are pleased to acknowledge support from ERC Advanced Grant 740900 (LogCorRM).  JPK was also supported by a Royal Society Wolfson Research Merit Award.

\end{document}